\documentclass[lettersize,journal]{IEEEtran}
\usepackage{amsmath,amsfonts}
\usepackage{algorithmic}
\usepackage{algorithm}
\usepackage{array}
\usepackage{textcomp}   
\usepackage{stfloats}
\usepackage{url}
\usepackage{verbatim}
\usepackage{graphicx}
\usepackage{cite}
\usepackage{color}

\usepackage{amsthm}
\newtheorem{theorem}{Theorem}

\usepackage[caption=false, font=footnotesize]{subfig}

\usepackage{footnote}

\allowdisplaybreaks[4]

\begin{document}

\title{FLARE: A New Federated Learning Framework with Adjustable Learning Rates over Resource-Constrained Wireless Networks}

\author{Bingnan Xiao,
        Jingjing~Zhang,~\IEEEmembership{Member,~IEEE}, 
        Wei~Ni,~\IEEEmembership{Fellow,~IEEE}
        and~Xin~Wang,~\IEEEmembership{Fellow,~IEEE}
\thanks{Bingnan Xiao, Jingjing Zhang, and Xin Wang are with the Department of Communication Science and Engineering, 
School of Information Science and Technology, Fudan University, 
Shanghai 200433, China (e-mail: 22110720061@m.fudan.edu.cn, \{jingjingzhang, xwang11\}@fudan.edu.cn).

Wei Ni is with Data61, CSIRO, Marsfield, NSW 2122, Australia, and the School of Computing Science and Engineering, the University of New South Wales (UNSW), Kennington, NSW 2052, Australia (e-mail:
Wei.Ni@data61.csiro.au).
}}

\maketitle

\begin{abstract}
Wireless federated learning (WFL) suffers from heterogeneity prevailing in the data distributions, computing powers, and channel conditions of participating devices. 
This paper presents a new \textbf{F}ederated \textbf{L}earning with \textbf{A}djusted
lea\textbf{R}ning rat\textbf{E} (FLARE) framework to mitigate the impact of the heterogeneity. 
The key idea is to allow the participating devices to adjust their individual learning rates and local training iterations, adapting to their instantaneous computing powers. 
The convergence upper bound of FLARE is established rigorously under a general setting with non-convex models in the presence of non-i.i.d. datasets and imbalanced computing powers. 
By minimizing the upper bound, we further optimize the scheduling of FLARE to exploit the channel heterogeneity. A nested problem structure is revealed to facilitate iteratively allocating the bandwidth with binary search and selecting devices with a new greedy method. A linear problem structure is also identified and a low-complexity linear programming scheduling policy is designed when training models have large Lipschitz constants.
Experiments demonstrate that FLARE consistently outperforms the baselines in test accuracy, and converges much faster with the proposed scheduling policy.
\end{abstract}


\section{Introduction}

The upcoming sixth-generation (6G) wireless communication systems will be characterized by pervasive connectivity. The volume of data generated by edge devices—spanning smartphones, wireless sensors, and wearable devices is expected to experience exponential growth \cite{9349624}, \cite{9446488}. This surge in data production will catalyst widespread adoption and proliferation of artificial intelligence (AI)~\cite{8648462}, and give rise to new federated learning (FL) techniques~\cite{konečný2018federated}. 
In a standard FL algorithm, namely FedAvg~\cite{fedavg}, a fixed number $\tau$ of local stochastic gradient descent (SGD) updates are performed at a selected set of devices in each training round before the updated parameters are sent to a server, e.g., a base station (BS), for aggregation. This process repeats until an adequate global model is obtained.

As an extension of FL, wireless FL (WFL) supports data processing and model training in wireless networks \cite{9141214}, and has been applied to, e.g., online path control for massive unmanned aerial vehicle (UAV) networks \cite{9169921}, distributed localization \cite{9250516}, and content recommendations for mobile devices~\cite{ammad2019federated}.
While FedAvg has demonstrated good empirical performance \cite{nilsson2018performance}, it requires some desirable conditions as a prerequisite, which is hardly possible in real-world wireless networks. 
    {One challenge for directly applying FedAvg in WFL is
    non-independent and identically distributed (non-i.i.d.) datasets among devices, also known as data heterogeneity.} The data heterogeneity tends to induce significant performance degradation \cite{li2019convergence}.
    Another challenge is that most FL frameworks only allow the participants to 
    perform the same number of local SGD iterations per communication round at a consistent training speed. 
    However, various edge devices may differ substantially in hardware specifications (e.g., CPU, memory, battery, etc.), 
    resulting in different training speeds and local training iteration numbers \cite{li2020federated}. 
    Moreover, the transmission capability, such as transmit power, varies among devices in many wireless systems.  
    FL algorithms, typically developed under the assumption of random uniform device selection, may not be suitable, especially when the system has a limited bandwidth. Devices with poor channel conditions can be prevented from contributing to FL training \cite{9579038}.

\vspace{-3mm}
\subsection{Related Work}
Recent studies have devoted significant attention to the heterogeneity of FL.
As for data heterogeneity,
the authors of \cite{zhao2018federated} enhanced the training performance by assigning a globally shared dataset across edge devices. 
In~\cite{scaffold}, a variance reduction algorithm called Scaffold was developed with local correction terms appended to mitigate the gradient drift.
The authors of \cite{li2021model} utilized contrastive learning at a model level to rectify local updates.
Yet, these works overlooked the device heterogeneity by assuming a consistent training speed across devices. Moreover, data sharing \cite{zhao2018federated} and correction terms \cite{scaffold} would increase CPU and memory burden, making it hard to deploy FL at low-power edge devices.

To mitigate device heterogeneity, 
FedProx was proposed in \cite{fedprox} by adding a proximal term to local training, where the weight of the proximal term needs to be carefully tuned for different training tasks. 
It was shown in \cite{fednova} that device heterogeneity could cause objective inconsistency. Then, FedNova was designed to alleviate such adverse impacts to some extent.
In \cite{fedlin}, FedLin was proposed with correction terms similar to stochastic variance reduction gradient (SVRG); yet, the additional gradient computation and transmission was non-negligible. In general, the above works \cite{fednova}, \cite{fedlin} relied on full or uniformly random device selection, and their convergence under arbitrary or unbalanced device selection cannot be guaranteed. This would hinder their application in wireless edge networks, where devices often need to be scheduled based on their instantaneous or statistical channel conditions, as well as computational and communication abilities.

To tackle channel heterogeneity, efforts have been directed toward achieving efficient aggregation and resource scheduling.
Recently, a new analog aggregation technique named over-the-air computation (OTA) has been utilized in FL to alleviate the bandwidth crunch from orthogonal transmission by exploiting the feature of multiple-access channel superposition \cite{XIAO2024}.
In \cite{our}, an optimal channel-adaptive power control strategy for OTA-FL was investigated.
For digital schemes, the authors of \cite{8737464} minimized the weighted sum of computation delay and energy consumption in each communication round under full device participation. In \cite{8761315}, a greedy strategy was studied to schedule as many devices as possible per round. Neither of these policies provided provable convergence. 
In the absence of communication resource constraints, the authors of \cite{8851249} compared three scheduling policies: random scheduling, round-robin, and proportional fairness. A heuristic scheduling policy was designed in \cite{9337227} with channel and update importance jointly considered. 
In \cite{10041216}, a channel- and data-aware scheduling strategy was developed under asynchronous aggregation, where only coarse-grained (e.g., evenly distributed) bandwidth allocation was considered.
Joint device scheduling and resource allocation were studied in \cite{9207871} to improve the trade-off between training rounds and participating devices.
By assuming full gradient descent (GD) at the devices, a joint optimization problem for device selection and bandwidth allocation was considered in \cite{9210812}. By contrast, gradient divergence was used in \cite{9170917} to define the importance of updates and develop a scheduling strategy based on such importance and channel equality.
Considering data and device heterogeneity, the authors of \cite{9796935} optimized the sampling strategy in the ideal lossless transmission environments.

\subsection{Contribution and Organization}

In this paper, we design a new scheduling framework to address the device, data and channel heterogeneity in WFL, where the learning rates of participating devices can be adjusted online, adapting to the instantaneous availability of their computing power.
A new device selection and bandwidth allocation strategy is developed to facilitate the convergence of WFL under the new framework.
The main contributions of the paper are summarized as follows.
    \begin{itemize}
        \item We propose a new framework, \textbf{F}ederated \textbf{L}earning with \textbf{A}djusted lea\textbf{R}ning rat\textbf{E} (FLARE), to address device heterogeneity in WFL. By adaptively adjusting the learning rates of the participating devices, FLARE allows for consistent training progress among the participating devices with substantially different computing powers, hence accelerating model convergence.

        \item Under FLARE, we develop a general convergence analysis of WFL with non-convex models in the presence of data and device heterogeneity, and arbitrary device scheduling strategy.

        \item Given the convergence upper bound, we design a new WFL scheduling policy under the FLARE framework, when the communication bandwidth is limited. To tackle the non-convexity of the convergence bound and the coupling of device selection and bandwidth allocation, we reveal a nested problem structure and decouple the problem to iteratively allocate the bandwidth with binary search and select devices with a new greedy strategy.

        \item We further reveal a linear problem structure of device selection and bandwidth allocation for models with large Lipschitz constants. A low-complexity linear programming based policy is designed accordingly.

    \end{itemize} 
Experiments demonstrate the effectiveness of FLARE in accelerating model convergence under different data distributions and scheduling strategies. 
Under the FLARE framework, the proposed scheduling policy outperforms the state-of-the-art strategies even without the learning rate adjustments and exhibits robustness under various system parameter settings.

The rest of this paper is organized as follows. Section~II presents the system model. Section III elaborates on the FLARE framework tailored to address device heterogeneity, analyzes the communication of FLARE under data heterogeneity, and formulates the problem to facilitate convergence. Section IV reformulates the problem and presents efficient algorithms.
Numerical results are provided in Section V, followed by conclusions in Section VI.

\textit{Notation:}
Calligraphic letters represent sets. Boldface lowercase letters indicate vectors. \(|\mathcal{A}|\) denotes the cardinality of set \(\mathcal{A}\), and $\setminus$ denotes the set difference operation. \(\mathbb{R}\) denotes the real number field. \(\mathbb{E}[\cdot]\) denotes statistical expectation. \(\nabla\) stands for gradient. 
\( \|\boldsymbol{\rm{w}}\|_2 \) is the \(\ell_2\)-norm of vector \(\boldsymbol{\rm{w}}\), \(\langle \boldsymbol{\rm{a}},\boldsymbol{\rm{b}} \rangle \) represents the inner product operation, and $\mathrm{mean} \{ \cdot \}$ denotes a mean-value function.

\section{System Model}
In this section, we elucidate the system model, where the participating devices can have unbalanced data, computing, and communication abilities.
\subsection{Federated Learning Model}
The considered WFL system comprises a BS and $K$ edge devices collected by $\mathcal{K} \buildrel \Delta \over =  \{1,2,\cdots, K \}$. The objective of WFL is to minimize the empirical loss function on a given dataset, as given by 
\begin{equation} \label{loss func}
 F(\boldsymbol{\rm{w}}) := \sum_{i=1}^{K}c_if_i(\bf{w}), 
\end{equation}
where ${\bf{w}} \in \mathbb{R}^{d}$ is the $d$-dimensional model parameter; $f_i(\boldsymbol{\rm{w}}) \buildrel \Delta \over = \mathbb{E}_{\xi_i \sim \mathcal{D}_i}\left[f_i(\boldsymbol{\rm{w}},\xi_i)\right]$ represents the local loss function concerning the local dataset $\mathcal{D}_i$ of device $i$; and $c_i$ denotes the aggregation weight of device $i$, satisfying $\sum_{i=1}^{K}c_i=1$.

In this paper, we study a generic scenario where heterogeneity prevails in the data, computing and communication conditions of devices.
Specifically, the data may not be i.i.d. among the devices; the devices may have substantially different computing powers and hence train their local models at different paces; and the communication from the devices to the BS is constrained by limited bandwidth and undergoes non-negligible latency.
Define a binary vector $\boldsymbol{a}:=[a_{r,1},\cdots,a_{r,K}]^{\top}$, where $a_{r,i} \in \{0,1\} $ indicates whether device $i$ is scheduled ($a_{r,i}=1$) or not ($a_{r,i}=0$) in round $r$.

At each training round $r$, the following steps are executed, as illustrated in Fig. \ref{fig:update_flow}:
\begin{itemize}
    \item Device scheduling and downlink transmission:
    The BS selects a subset of devices, denoted by $\mathcal{M}_r$, with $M_r$ devices, and multicasts the global parameter $\boldsymbol{\rm{w}}_r$ to every selected device $i \in \mathcal{M}_r$. 

    \item Local update: By setting $\boldsymbol{\rm{w}}_{r,i}^{0}=\boldsymbol{\rm{w}}_r$, each selected device $i$ begins $\tau_{r,i}$ local updates according to its local computational capacity. The updating rule for SGD is 
    \begin{equation}
    \boldsymbol{\rm{w}}_{r,i}^{j+1}  = \boldsymbol{\rm{w}}_{r,i}^{j} - \eta_{\mathrm{l}} g_{r,i}^{j}, \,\,\, j=0,1,\cdots,\tau_{r,i}-1,
    \end{equation}
    where $\eta_{\mathrm{l}}$ denotes the local learning rate of the devices, and $g_{r,i}^j = \nabla f_i(\boldsymbol{\rm{w}}_{r,i}^j,\xi_{r,i} )$ is the stochastic gradient of device $i$ with respect to the mini-batch  $\xi_{r,i}$ with size $D$, sampled from the local dataset $\mathcal{D}_i$. 

    \item Uploading and aggregation: After local computation, each selected device $i$ uploads its cumulative local gradient $\Delta_{r,i} = \boldsymbol{\rm{w}}_{r,i}^{\tau_{r,i}}-\boldsymbol{\rm{w}}_{r,i}^{0}$. With the global learning rate $\eta_{\rm g}$, the BS updates the global parameter 
    $\boldsymbol{\rm{w}}_{r+1}$, as follows.
    \begin{align}
    \boldsymbol{\rm{w}}_{r+1}  &=\boldsymbol{\rm{w}}_r + \eta_{\rm g}\sum_{i \in \mathcal{M}_r }c_i\Delta_{r,i}  \nonumber \\
    &= \boldsymbol{\rm{w}}_r - \eta_{\rm g} \sum_{i \in \mathcal{M}_r} \frac{1}{M_r} \sum_{j=0}^{\tau_{r,i}-1}\eta_{\mathrm{l}} g_{r,i}^j.
    \end{align}

\end{itemize}
Then, the next $(r+1)$-th training round starts. This repeats until the prespecified termination criteria are met.

\begin{figure}[!t]
	\centering
	\includegraphics[width=0.9\linewidth]{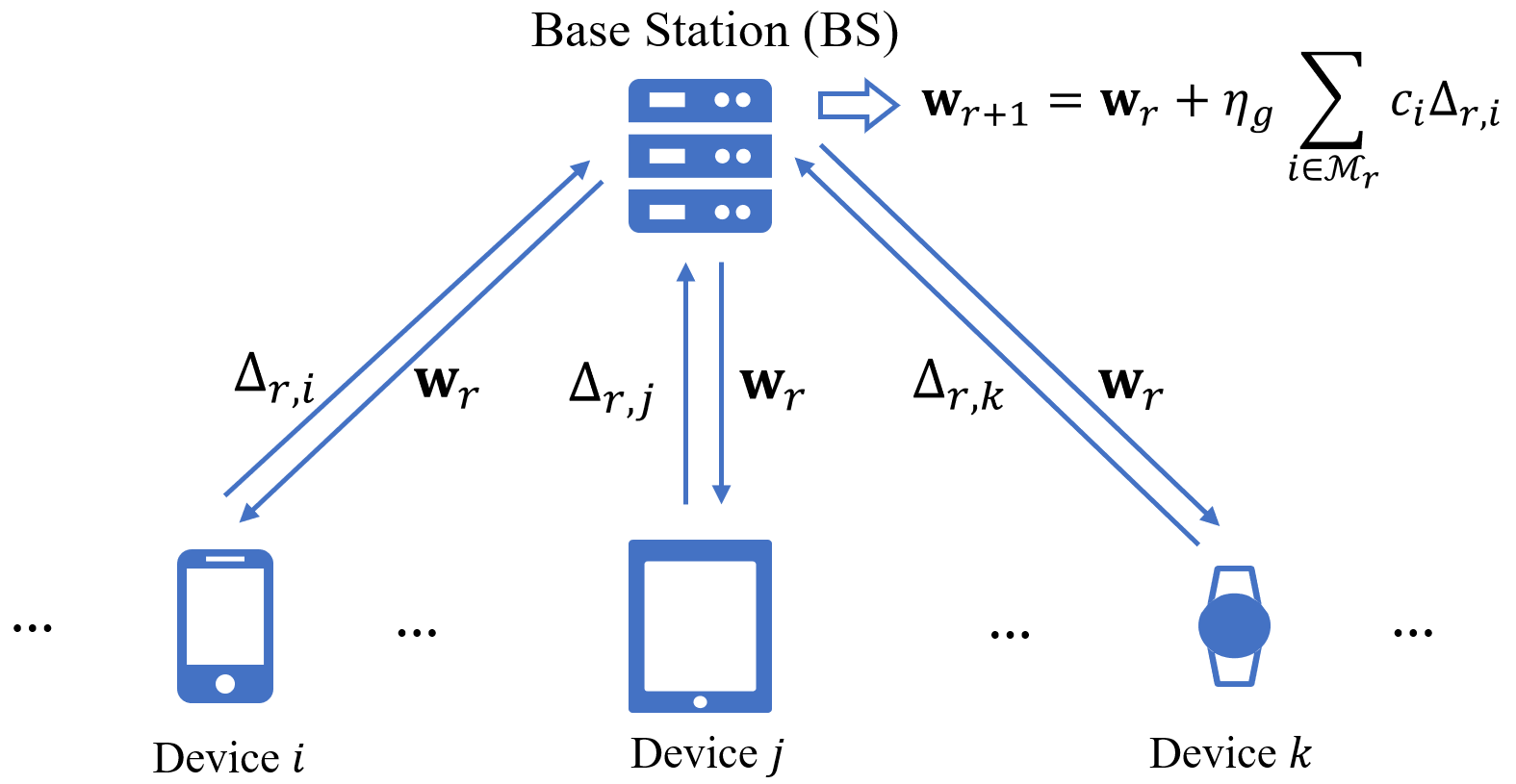}
	\caption{The architecture of a WFL system with $M_r$ selected devices at round~$r$.}
	\label{fig:update_flow}
\end{figure}

\subsection{Local Computation-Communication Latency Model}
For each round $r$, the total latency $t_{r,i} = t_{r,i}^{\mathrm{comp}} + t_{r,i}^{\mathrm{comm}}$ undergone by each selected device $i$ consists of the local computation delay $t_{r,i}^{\mathrm{comp}}$ and transmission delay $t_{r,i}^{\mathrm{comm}}$.
\subsubsection{Local Computation}
Let $f_{r, i} $ represent the computation capacity of device $i$ in round $r$, measured in the number of CPU cycles per second. The computation time needed for the $\tau_{r,i}$ local training iterations at device $i$ is given by
\begin{equation}
    t_{r,i}^{\mathrm{comp}} = \frac{\tau_{r,i}C_iD}{f_{r,i}},
\end{equation}
where $C_i$ (cycles/sample) is the number of CPU cycles required to compute a data sample. Assume that the computation latency needed for model aggregation at the BS is relatively negligible due to its abundant computational capacity and the considerably low complexity of the aggregation process.

\subsubsection{Wireless Transmission}
After local computation, the selected devices upload their local updates via frequency-division multiple access (FDMA) with a total bandwidth $B$.
We consider a line of sight (LoS)-dominating environment, since future wireless systems will predominantly operate in mmWave/THz with quasi-light propagations.
For device $i$, the achievable uplink rate is given by 
\begin{equation}
    {u_{r,i}} = {b_{r,i}}{\log _2}\left( {1 + \frac{{{p_{r,i}}h_{r,i}^2}}{{{b_{r,i}}{N_0}}}} \right),
\end{equation}
where $b_{r,i}$ is the bandwidth allocated to device $i$, satisfying $\sum_{i \in \mathcal{K}}a_{r,i}b_{r,i} \leq B$; $p_{r,i} $ is the transmit power of device $i$ in round $r$; $h_{r,i}$ denotes the corresponding channel amplitude gain; $N_0$ is the power spectral density (PSD) of the additive white Gaussian noise (AWGN) at the BS. 

Thus, the transmission latency of device $i$ is given by
\begin{equation}
    t_{r,i}^{\mathrm{comm}} = \frac{S}{u_{r,i}},
\end{equation}
where $S$ is the size of $\Delta_{r,i}$ (in bits). 
We assume that the transmissions between the BS and devices are reliable and free of packet errors.
Consider synchronous aggregation, then the total latency per round is determined by the slowest of the selected devices. Thus, we require
\begin{equation}
\max _{i \in \mathcal{K}}\left\{a_{r,i} t_{r,i}\right\} = \max _{i \in \mathcal{K}}\left\{a_{r,i}\left(t_{r, i}^{\mathrm{comp}}+t_{r, i}^{\mathrm{comm}}\right)\right\} \leq t_{\mathrm{thr}},
\end{equation}
where $t_{\mathrm{thr}}$ denotes the (maximum allowable) overall latency threshold per round.

\section{Problem Formulation of federated learning with dynamically adjusted learning rates}

In this section, we propose the FLARE framework, which mitigates the impact of inconsistent local training iterations between devices and expedites the convergence by dynamically adjusting the learning rates of devices involved. 

\begin{figure}[!t]
	\centering
	\subfloat[]{\includegraphics[width=1.69in]{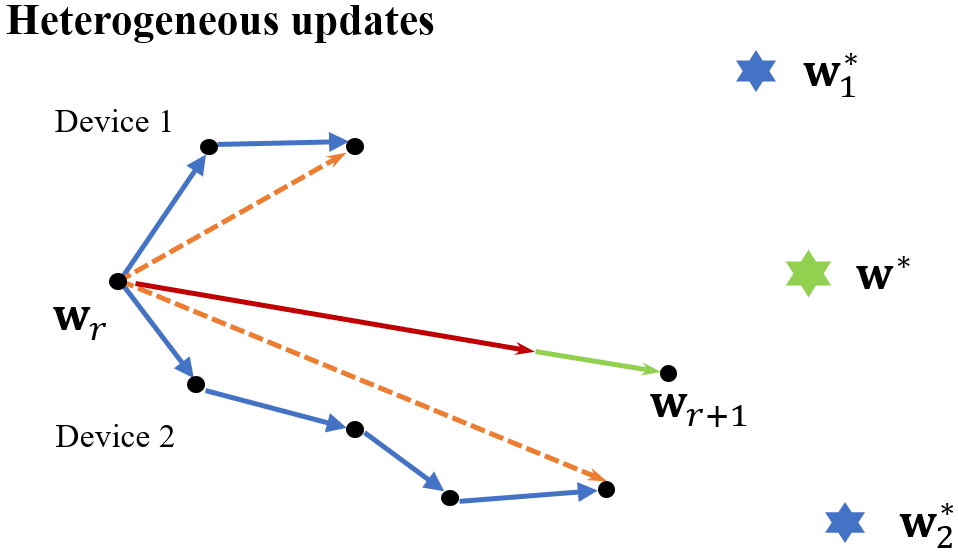}
		\label{fig_2_1}}
	\subfloat[]{\includegraphics[width=1.69in]{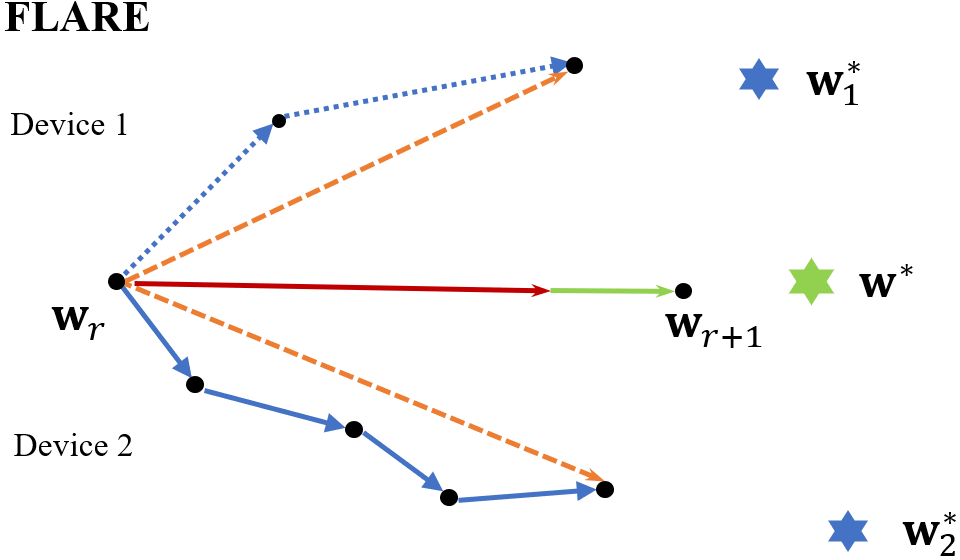}
		\label{fig_2_2}}
	\caption{An illustration of model updates of a heterogeneous setting with (a) equal learning rates and (b) FLARE. The green and blue marks represent the minima of global and local objectives, respectively.}
	\label{fig_2}
\end{figure}

\subsection{Overview of FLARE}
Within the FL architecture, the significant difference in the training loss landscape of different local models, resulting from the distinct learning capabilities of different devices, can lead to a global objective deviation and hinder the overall convergence~\cite{diao2021heterofl,qu2022generalized}. 
As illustrated in Fig. 2(a), substantially imbalanced updating progresses between different devices (e.g., Device 2 updates its local model much faster than Device~1) are likely to lead to non-negligible deviation of the updated global model ${\bf{w}}_{r+1}$ from its optimal $\boldsymbol{\rm{w}}^*$.
It is imperative to enhance the consistency of local training progress towards the optimal global parameter $\mathbf{w}^*$ in a resource-constrained WFL system.

We propose to \textit{dynamically adjust the learning rates of the selected devices based on their computation capabilities} so that the local SGD updates of the devices can progress at consistent rates per round.
The local learning rates are linearly scaled to prevent the global objective deviation caused by unbalanced local updates~\cite{fednova}. 
In other words, resource-constrained devices can conduct a relatively smaller number of local training iterations at faster learning rates. As a result, they can make consistent progress in local training with their more computationally powerful peers, and avoid becoming stragglers, as illustrated in Fig. 2(b).

At any round $r$, each device reports its available instantaneous computational resources to the BS. This report incurs negligible communication overhead compared to the local model updates. After a device subset $\mathcal{M}_r$ is selected, the BS determines $\bar{\tau}_r$ for the selected devices to adjust their local training iteration numbers.
Specifically, the BS sends $\bar{\tau}_r$ and the global model ${\bf{w}}_{r}$ to the selected devices. 
By setting its learning rate to $\widetilde{\eta}_{\mathrm{l}}=\eta_{\mathrm{l}}\bar{\tau}_r/\tau_{r,i}$, device $i \in \mathcal{M}_r$ executes $\tau_{r,i}$ local SGD iterations per round $r$, as given by 
\begin{align} \label{adjusted update}
    {\bf{w}}_{r,i}^j 
    &= {\bf{w}}_{r,i}^{j-1} - \widetilde{\eta}_{\mathrm{l}} g_{r,i}^{j-1} \\ \nonumber 
    &= {\bf{w}}_{r,i}^{j-1} - \eta_{\mathrm{l}}\frac{\bar{\tau}_r}{\tau_{r,i}}g_{r,i}^{j-1},\,j=0,1,\cdots,\tau_{r,i}-1.
\end{align}
After local computations, each selected device $i$ uploads the gradient update $\widetilde{\Delta}_{r,i} = {\bf{w}}_{r,i}^{\tau_{r,i}}-{\bf{w}}_{r}$ to the BS.
The BS then performs the following aggregation:
\begin{align}
    \boldsymbol{\rm{w}}_{r+1} = \boldsymbol{\rm{w}}_r - \eta_{\rm g} \sum_{i \in \mathcal{M}_r} \frac{1}{M_r} \sum_{j=0}^{\tau_{r,i}-1} \eta_{\mathrm{l}}\frac{\bar{\tau}_r}{\tau_{r,i}}g_{r,i}^{j-1}.
\end{align}

\noindent\textbf{Remark 1:}
FLARE can be viewed as a generalization version of FedAvg by mitigating the imbalance in local SGD iterations. It reduces to FedAvg if $\mathcal{M}_r$ is randomly uniformly selected with $\tau_{r,1}=\dots=\tau_{r,K}$ and $\eta_{\rm g} = 1$ at each round $r$.

\subsection{Approximation Error Analysis}
We analyze the approximation errors caused by the proposed learning rate adjustment in FLARE, with respect to performing a consistent number of local training iterations (e.g., $\overline{\tau}_r$ iterations) across all devices.
We establish the following Theorem.

\begin{theorem}
Suppose that ${\bf{w}}$ is obtained by taking $\overline{\tau}_r$ SGD steps with a learning rate $\eta_{\mathrm{l}}$, and $\widetilde{\bf{w}}$ is obtained by taking ${\tau}_r$ SGD steps with a learning rate ${\widetilde \eta}_{\mathrm{l}} = \eta_{\mathrm{l}}{\overline{\tau}_r}/{\tau}_r$.
Starting from the same initial point with a non-convex function $f$, the mean squared error (MSE) of $\widetilde{\bf{w}}$ with respect to ${\bf{w}}$ is bounded by
\begin{equation} \label{approximation error}   
\mathbb{E}\left[\left\|\widetilde{\boldsymbol{\rm{w}}}-\boldsymbol{\rm{w}}\right\|_2^2 \right] \leq 8\eta_{\mathrm{l}}^2{\overline{\tau}_r^2}(g+\sigma^2), 
\end{equation}
\textit{where $\sigma^2$ and $g$ denote the SGD variance and gradient upper bound, respectively.}
\end{theorem}
\begin{proof}
    See Appendix A.
\end{proof}

From \textbf{Theorem 1}, the RHS of \eqref{approximation error} depends only on the constants, $g$ and $\sigma^2$, when the local learning rate $\eta_{\mathrm{l}}$ is inversely proportional to $\overline{\tau}_r$, i.e., $\eta_{\mathrm{l}} = \mathcal{O}\left(\frac{1}{\overline{\tau}_r}\right)$.   
As a consequence, the MSE, $\mathbb{E}\left[\left\|\widetilde{\boldsymbol{\rm{w}}}-\boldsymbol{\rm{w}}\right\|_2^2 \right]$, is bounded under FLARE.
Also, such approximation errors decrease with the progress of model training since the corresponding gradient paradigm becomes smaller, leading to a tighter gradient upper bound $g$.
In Section III-C, we demonstrate that the relationship between $\eta_{\mathrm{l}}$ and $\bar{\tau}_r$ plays a critical role in reducing the convergence upper bound.

\subsection{WFL Problem Formulation}
For WFL with constrained computing and communication resources, the objective is to minimize the training loss. A general formulation of such a problem is given by
\begin{subequations}\label{overall problem}
    \begin{align}	~~~~~~\textbf{P1}:~~~&\min_{\boldsymbol{a},\boldsymbol{b}}\quad F({\bf{w}}) \label{ini obj}\\
		 \text{s.t.}~
            & ~~a_{r,i} \in \{ 0,1 \}  , \label{constraint 0}\\
		& ~~\sum_{i \in \mathcal{K}}a_{r,i}b_{r,i} \leq B  , \label{constraint 1}\\
		& ~~a_{r,i}t_{r,i} \leq t_{\mathrm{thr}} ,\forall i \in \mathcal{K},\forall r \in [R]. \label{constraint 2}
	\end{align}
\end{subequations}
where $\boldsymbol{b} = [b_{r,1},\cdots,b_{r,K}]^{\top}$ collects the bandwidth allocations of the $K$ devices in round $r$.
Constraint \eqref{constraint 1} specifies that the total bandwidth $B$ must not be exceeded; and \eqref{constraint 2} dictates that the computation and transmission must be completed before a proposed latency threshold $t_{\mathrm{thr}}$ in each round. 

During the training process, the BS dynamically determines a subset of devices and allocates bandwidths to achieve the minimal global loss while satisfying constraints \eqref{constraint 1} and \eqref{constraint 2}.
In order to analyze how the convergence of FLARE is affected by device selection and resource allocation, we next establish an asymptotic convergence upper bound for $F(\bf{w})$ and minimize the derived upper bound to seek an efficient solution for \eqref{overall problem}. 

\subsection{Convergence Analysis}
We analyze the convergence of FLARE, which contributes to the improved tractability of Problem \textbf{P1}, as will be described in Section IV. 
This starts with the following assumptions.

\noindent \textbf{Assumption 1.} (Unbiasedness and Variance Boundedness).
For a model parameter $\mathbf{w}$, the stochastic gradient of each device is an unbiased estimator of the true local gradient, i.e.,
\begin{equation}
   \mathbb{E}[\nabla f_i(\mathbf{w} ; \xi)]=\nabla f_i(\mathbf{w}), \,\,\, \forall i \in \mathcal{K}.
\end{equation}
Moreover, there exists a bounded SGD variance $\sigma^2$ such that
\begin{equation}
    \mathbb{E}\left[\|\nabla f_i(\mathbf{w} , \xi)-\nabla f_i(\mathbf{w})\|^2\right] \leq \sigma^2, \,\,\, \forall i \in \mathcal{K}.
\end{equation}

\noindent \textbf{Assumption 2.} (Smoothness)
The local objective function $f_i$ is Lipschitz smooth; i.e., there exists a constant $L>0$ satisfying
\begin{equation}
    \left\|\nabla f_i(\boldsymbol{\rm{w}}_2)-\nabla f_i(\boldsymbol{\rm{w}}_1)\right\| \leq L\|\boldsymbol{\rm{w}}_2-\boldsymbol{\rm{w}}_1\|, \,\,\, \forall i \in \mathcal{K}.
\end{equation}

\noindent \textbf{Assumption 3.} (Gradient Dissimilarity Boundedness) There exist two constants $G \geq 0$ and $H \geq 1$ such that
\begin{equation}\label{BGD}
    \left\|\nabla f_i(\boldsymbol{\rm{w}})\right\|^2 \leq G^2+H^2\|\nabla F(\boldsymbol{\rm{w}})\|^2, \,\,\, \forall i \in \mathcal{K}.
\end{equation}

\textbf{Assumptions 1 -- 3} have been extensively considered in the literature, e.g., \cite{fednova,fedlin,scaffold}.
Under these assumptions, we analyze the one-round convergence of FLARE as follows.

\noindent\textbf{Lemma 1.} \textit{For each round $r$ and the subset of selected devices $\mathcal{M}_r$, the upper bound of the one-round convergence of FLARE with a non-convex loss function is given by } 
\begin{equation} \label{one round error}
    \mathbb{E}\!\left[ F(\!{\bf{w}}_{r+1}\!)\!-\!F(\!{\bf{w}}_{r}\!)\right] 
    \!\leq \!\! -\frac{\bar{\tau}_r \eta_{\rm g} \eta_{\mathrm{l}} \mathrm{q}_{r}}{2}\!\left\|\nabla \! F\!\left({\bf{w}}_r\right)\right\|^2 \!\!+ \!\phi_1 \!\!+ \!\phi_2 ,
\end{equation}
\textit{where $\phi_1 = 1.2L^2\bar{\tau}_r^3{\eta_{\rm g}\eta_{\mathrm{l}} ^3}G^2$, $\phi_2 = \frac{L \bar{\tau}_r^2 \eta_{\rm g} \eta_{\mathrm{l}}^2}{2}\left(\frac{ L \bar{\tau}_r \eta_{\mathrm{l}}}{ M_r}+\frac{\eta_{\rm g}}{M_r^2}\right) \sigma^2 \sum_{i \in \mathcal{M}_r} \frac{1}{\tau_{r, i}}$, and $\mathrm{q}_{r}$ is a round-specific positive constant.}
\begin{proof}
    See Appendix B.
\end{proof}

\textbf{Lemma 1} reveals that the convergence bound of any round is affected by $\phi_1$ and $\phi_2$.
Here, $\phi_1$ is determined by $G^2$, which reflects the degree of inconsistency between the local and global gradients, as stated in \textbf{Assumption 3}.
By contrast, $\phi_2$ can be further controlled as it depends heavily on the numbers of selected devices $M_r$ and local training iterations $\tau_{r,i}$. We note that $\phi_2$ can be reduced by selecting an adequate $M_r$, and scheduling each device to perform a sufficient number $\tau_{r, i}$ of local updates in round $r$. 
This approach is more effective in decreasing the upper bound $\mathbb{E}\left[ F({\bf{w}}_{r+1})-F({\bf{w}}_{r})\right]$ than only maximizing the number of selected devices (as done in, e.g.,\cite{8761315,9237168}).

Given the one-round convergence in \textbf{Lemma 1}, we elucidate the convergence performance of FLARE in \textbf{Theorem 2}. 
For illustration convenience, we assume $\bar{\tau}_r = \bar{\tau},\, \forall r = 1,\cdots,R$ in the theorem. (Nevertheless, 
it is straightforward to extend \textbf{Theorem 2} to the situation where $\bar{\tau}_r,\,\forall r$ changes over rounds, i.e., by applying additional round-specific scaling to \eqref{septal terms cancel 2} before summing \eqref{septal terms cancel} up from $r=1$ to $R$ in Appendix C.) 

\begin{theorem}\label{Convergence}
    Given an initial global model ${\bf{w}}_1$ and the optimal global model ${\bf{w}}^*$, the convergence upper bound of FLARE with a total of $R$ training rounds is given by
    \begin{align}\label{non-convex conv}
     \frac{1}{R}\sum_{r=1}^{R} \mathbb{E}\left[||\nabla F({{\bf{w}}_r})|{|^2}\right] \leq \frac{2[F({\bf{w}}_1)-F({\bf{w}}^*)]}{\bar{\tau}\eta_{\rm g}\eta_{\mathrm{l}}\mathrm{q} R} +\Phi_1+\Phi_2 , 
    \end{align}
where $\Phi_1 =  \frac{2.4}{\mathrm{q}}L^2\bar{\tau}^2\eta_{\mathrm{l}}^2G^2 $, $\Phi_2 =  \frac{L\eta_{\mathrm{l}}}{\mathrm{q}}\left( L \bar{\tau} \eta_{\mathrm{l}}+ \frac{\kappa \eta_{\rm g}}{M}\right) \sigma^2$
, $\kappa = \mathop {\max }\limits_{r,i} \{ \frac{{{\overline{\tau}}}}{{{\tau_{r,i}}}}\}$, $M  \buildrel \Delta \over = \mathop {\min }\limits_r \{M_r\}$, and $\mathrm{q}$ is a positive constant. 

\end{theorem}

\begin{proof}
    See Appendix C.
\end{proof}

Compared with the existing studies 
\cite{zhao2018federated,scaffold,li2021model}, the convergence analysis of FLARE in \textbf{Theorem 2} is generic in the sense that it supports both non-i.i.d. data and imbalanced local training, i.e., data and device heterogeneity with non-convex training models. 

\begin{itemize}
    \item By observing the RHS of \eqref{non-convex conv} in comparison with \eqref{septal terms cancel 3}, we can reveal the impact of device selection on the convergence of FLARE: In any round $r$, a larger $M_r$ leads to a tighter upper bound in \eqref{septal terms cancel 3} since $M  \buildrel \Delta \over = \mathop {\min }\limits_r \{M_r\}$, where $M$ denotes the minimum number of selected devices in $R$ rounds. 
    Moreover, $\frac{\kappa L\eta_{\mathrm{l}}\eta_{\rm g}}{\mathrm{q} {M}}$ in $\Phi_2$ decreases linearly with the number $M$ of selected devices. Under full participation at each round, i.e., $M=K$, the convergence gap is the minimum.

    \item 
    The parameter $\kappa$ in $\Phi_2$ represents the maximal difference of local training iterations relative to $\overline{\tau}$, which illustrates the convergence error introduced by the imbalanced local SGD iterations.
    It would decrease if the selected devices perform a consistent number of local iterations, i.e., $\tau_{r,i}$ is equal, $\forall i$. The corresponding error $\frac{L\eta_l\kappa\eta_{\rm g}}{\mathrm{q}M}\sigma^2$ in $\Phi_2$ can be reduced if $M$ is larger.

    \item 
    As shown in \eqref{non-convex conv},
    $\Phi_1$ and $\Phi_2$ do not vanish with the increase of training rounds $R$. They have similar structures and grow with $G^2$ and $\sigma^2$.
    To effectively bound ${\bar{\tau}}^2\eta_{\mathrm{l}}^2$ in $\Phi_1$ and $\bar{\tau}\eta_{\mathrm{l}}^2$ in $\Phi_2$, 
    we design $\eta_{\mathrm{l}} = \mathcal{O}\left(\frac{1}{\bar{\tau}}\right)$. 
    As a result, both $\Phi_1$ and $\Phi_2$ are upper bounded in the sense that they do not grow with $\bar{\tau}$.
    In turn, the approximation error in \eqref{approximation error} is upper bounded, as described in \textbf{Theorem 1}.
    
\end{itemize}

We also establish the convergence rate of FLARE through an appropriate configuration of global and local learning rates, as delineated in the following corollary.

\vspace{2 mm}
\noindent \textbf{Corollary 1}\textbf{.}
    \textit{Let $\eta_{\mathrm{l}}=\frac{1}{\bar{\tau}\sqrt{R}}$ and $\eta_{\rm g}=\sqrt{\bar{\tau}M}$. The convergence rate of FLARE is given by}
    \begin{equation} \label{O-form}
        \frac{1}{R}\sum_{r=1}^{R} \mathbb{E}\left[||\nabla F({{\bf{w}}_r})|{|^2}\right] \leq \mathcal{O}\left( \frac{1}{\sqrt{\bar{\tau}MR}} + \frac{1}{R} \right).   
    \end{equation}
\begin{proof} 
By plugging $\eta_{\mathrm{l}}=\frac{1}{\bar{\tau}\sqrt{R}}$ and $\eta_{\rm g}=\sqrt{\bar{\tau}M}$ into \eqref{non-convex conv} and suppressing 
the low-order terms,
\eqref{O-form} is obtained.
\end{proof}

\noindent \textbf{Remark 2:} It is indicated in \textbf{Corollary 1} that FLARE with the local and global learning rates of $\eta_{\mathrm{l}}=\frac{1}{\bar{\tau}\sqrt{R}}$ and $\eta_{\rm g}=\sqrt{\bar{\tau}M}$, can achieve a linear convergence speedup with the number $M$ of devices\footnote{To attain the accuracy of $\epsilon$, $\mathcal{O}\left(\frac{1}{\epsilon^2}\right)$ steps are needed with a convergence rate $\mathcal{O}\left(\frac{1}{\sqrt{R}}\right)$, while $\mathcal{O}\left(\frac{1}{M \epsilon^2}\right)$ steps are needed if the convergence rate is $\mathcal{O}\left(\frac{1}{\sqrt{M R}}\right)$. In this sense, we achieve a linear speedup with increasingly selected devices.}. The convergence rate yields $\mathcal{O}\left( \frac{1}{\sqrt{\bar{\tau}MR}} \right)$ when $R>\sqrt{\bar{\tau}M}$. 
To the best of our knowledge, this is the first work to achieve such a convergence rate 
under non-i.i.d. data, imbalanced local updates (e.g., in terms of local iteration number per round), and arbitrary device selections in a general non-convex case, thereby validating the effectiveness of FLARE.

\vspace{2 mm}
\noindent \textbf{Remark 3:}
Compared to the existing methods, e.g., \cite{yang2021achieving, scaffold, li2019convergence}, FLARE is advantageous in both convergence rate and communication overhead. In non-convex cases, the convergence rate of FLARE, i.e., $\mathcal{O}\left( \frac{1}{\sqrt{\bar{\tau}MR}} \right)$, surpasses those of prior works, e.g., $\mathcal{O}\left( \frac{\bar{\tau}}{\sqrt{MR}} \right)$ in \cite{yang2021achieving,li2019convergence} by properly designing the learning rates. 
The convergence rate of FLARE increases with the number $\bar{\tau}$ of SGD iterations per round.
In variance reduction methods, e.g., Scaffold~\cite{scaffold}, devices perform local updates with correction terms of the same dimension as the model parameter ${\bf{w}}$. The BS aggregates these terms to update the global model.
FLARE achieves a comparable convergence rate, while device $i$ only needs to report $\tau_{r,i}$ to the BS, which is relatively negligible compared to the correction terms used in variance reduction methods.

\section{Proposed Scheduling under FLARE}
In light of \eqref{one round error} and \eqref{non-convex conv}, we optimize bandwidth allocation and device selection by minimizing the convergence upper bound. We first develop an iterative algorithm that decouples bandwidth allocation and device selection to alleviate communication bottlenecks and accelerate convergence. 
We further reveal a convex structure of Problem \textbf{P2} under certain scenarios and propose an efficient solution with linear programming. 

\subsection{Problem Reformulation}
Minimizing the one-round convergence upper bound in \eqref{one round error} contributes to decreased convergence upper bound in \textbf{Theorem~2}, i.e., \eqref{non-convex conv}.
Consequently, the minimization of $F(\bf{w})$ can be approximated by minimizing the RHS of \eqref{one round error}, which can be efficiently achieved by minimizing $\phi_2$. Problem \textbf{P1} can be rewritten as
\begin{subequations}\label{P2}
\begin{align}	
    \textbf{P2}: ~~&\min_{\boldsymbol{a},\boldsymbol{b}} ~
     \left( \!\frac{1}{M_r}+\frac{\gamma}{M_r^2}\!\right)\!\sum_{i \in \mathcal{M}_r}\frac{1}{\tau_{r,i}} , \label{obj of P2} \\
    \text{s.t.}
            & \quad a_{r,i} \in \{ 0,1 \}  , \label{constraint 1 of P2}\\
		& \quad \sum_{i \in \mathcal{K}}a_{r,i}b_{r,i} \leq B  , \label{constraint 2 of P2}\\
		& \quad a_{r,i}t_{r,i} \leq t_{\mathrm{thr}} ,\forall i \in \mathcal{K}, \,\,\, \forall r = 1,\cdots,R. \label{constraint 3 of P2}
\end{align}
\end{subequations}
where we define $\gamma = \frac{\eta_{\rm g}}{L}$ since $\eta_{\mathrm{l}} \sim \mathcal{O}\left(\frac{1}{\bar{\tau}}\right)$, according to \textbf{Theorem 1}.
Problem \textbf{P2} implies that the BS should optimally specify the number of scheduled devices and their respective local computation and transmission capabilities to minimize \eqref{obj of P2}, while adhering to the bandwidth and delay constraints \eqref{constraint 2 of P2} and \eqref{constraint 3 of P2}. Problem \textbf{P2} is a mixed integer nonlinear programming (MINLP) problem.
\subsection{Joint Bandwidth Allocation and Device Participation}

According to \eqref{constraint 3 of P2}, selecting devices with smaller delays helps admit more devices into an FL training round (i.e., increasing $\mathcal{M}_r$) and, in turn, decrease the objective function in \eqref{obj of P2}.
We can view the MINLP Problem \textbf{P2} as a nested problem by treating $t_{r,i},\,\forall i, r$ as a function of $\boldsymbol{a}$ and $\boldsymbol{b}$, albeit $t_{r,i}$ does not appear explicitly in \eqref{obj of P2}. This is because minimizing $t_{r,i}$ can decrease (19a) by expanding the feasible area of Problem \textbf{P2}, i.e., with (19d) satisfied.
For this reason, we decompose \textbf{P2} into two manageable sub-problems.
We solve the optimal bandwidth allocation strategy $\boldsymbol{b}^{*}$ for any given (feasible) select of selected devices $\mathcal{M}_r$ to minimize the total latency in round $r$. Then, we develop a greedy method to find $\mathcal{M}_r$ that minimizes (19a). These two steps iterate till convergence, achieving the overall policy. 

\subsubsection{Bandwidth Allocation} Given a device subset $\mathcal{M}_r$, we ensure the total latency of the selected devices satisfies \eqref{constraint 3 of P2} by carefully allocating the available bandwidth $B$ to the devices, reducing the overall latency per round. The bandwidth allocation is cast as the following min-max problem:  
\begin{subequations}\label{solution of b}
\begin{align}	
    \textbf{P3}:\quad &  \mathop {\min }\limits_{\boldsymbol{b}} \mathop {\max }\limits_i \quad t_{r,i} \label{obj of P3}\\
    \text{s.t.}
		& \quad \sum_{i \in \mathcal{M}_r}b_{r,i} \leq B  , \label{constraint 1 of P3}\\
  		& \quad b_{r,i} \geq 0,  \forall i \in \mathcal{M}_r,\,\,\, \forall r = 1,\cdots,R,
\end{align}
\end{subequations}
where \eqref{constraint 2 of P2} is simplified to \eqref{constraint 1 of P3} as $a_{r,i}=1, \forall i \in \mathcal{M}_r$. 

Let $t_r^* = \mathop {\min }\limits_{\boldsymbol{b}} \mathop {\max }\limits_i \ t_{r,i}$ denote the optimal value of \eqref{obj of P3}. 
Given $\mathcal{M}_r$, the optimal bandwidth allocation $\boldsymbol{b}^*$ and $t_r^*$ have the following relationship.

\begin{theorem}
The optimal bandwidth allocation of Problem \textnormal{\textbf{P3}} satisfies
\begin{equation} \label{optimal bandwidth}
    b_{r,i}^{*} = \frac{S\ln 2}{(t_r^*-t_{r,i}^{\mathrm{comp}})\left( \upsilon_{r,i} + W(-\upsilon_{r,i}e^{-\upsilon_{r,i}}) \right)},
\end{equation}
\textit{where ${\upsilon_{r,i}} =  \frac{SN_0\ln2}{p_{r,i} h_{r,i}^2 \left( t_r^*-t_{r,i}^{\mathrm{comp}}\right)} $, 
$W(\cdot)$ denotes the Lambert-W function, and $t_{r}^{*}$ satisfies}
\begin{equation} \label{optimal time}
    \sum_{i \in \mathcal{M}_r} b_{r,i}^* = B.
\end{equation}
\end{theorem}

\begin{proof}
    See Appendix D.
\end{proof}

Based on \eqref{optimal bandwidth} and \eqref{optimal time}, a binary search method can be employed to obtain $t_r^*$ numerically and $b_{r,i}^{*}, \forall i$ subsequently.

\subsubsection{Device Selection}
Given the fixed bandwidth allocation strategy $\boldsymbol{b}^{*}$, device selection can be written as 
\begin{subequations}\label{P4}
\begin{align}	
    \textbf{P4}:~~ \min_{\boldsymbol{a}} ~
    & \left( \!\frac{1}{M_r}+\frac{\gamma}{M_r^2}\!\right)\!\sum_{i \in \mathcal{M}_r}\frac{1}{\tau_{r,i}} \label{obj of P4} \\
    \text{s.t.}
            & \quad t_r^* \leq t_{\mathrm{thr}} ,\forall \mathcal{M}_r \subset  \mathcal{K}, \,\,\, \forall r = 1,\cdots,R. \label{constraint 1 of P4}
\end{align}
\end{subequations}

Even with the fixed $\boldsymbol{b}^{*}$, Problem \textbf{P4} is an MINLP. We develop an iterative greedy algorithm to solve Problem \textbf{P4} efficiently. 
In the $k$-th iteration of the greedy algorithm, we first specify an ``available'' subset $\mathcal{M}_r$ of candidate devices that decreases the objective of \eqref{obj of P4}. Then, we find the specific device that minimizes the increase of the total latency $t_r$ in the available subset solving Problem~\textbf{P3}.

To determine the available subset of devices, denoted by $\Pi_k$, for the BS in the $k$-th iteration of device selection, we come up with the following Theorem.

\begin{theorem}
    Assuming that a device subset $\mathcal{Q}_k$ with $Q$ devices has been selected for up to the $k$-th iteration, the BS shall continue to select devices from the set that satisfies the following condition:
\begin{equation} \label{descent condition}
\begin{aligned}
    \frac{1}{\tau_{r,i}} < \frac{Q^2+(2\gamma+1)Q+\gamma}{Q^2(Q+\gamma+1)}\sum_{q \in \mathcal{Q}_r}\frac{1}{\tau_{r,q}}, \,\,\, \exists \, i \in \mathcal{N},
\end{aligned}
\end{equation}
\textit{where $\mathcal{N}=\mathcal{K} \setminus \mathcal{Q}_k$.}
\end{theorem}
 
\begin{proof}
    See Appendix E.
\end{proof}
\textbf{Theorem 4} ensures that the selected device at each device selection iteration decreases the value of the objective function in \eqref{obj of P4}.
By employing \eqref{descent condition}, the BS does not need to evaluate all unselected devices at every step, thereby improving search efficiency. The size of the subset $\Pi_k$ also decreases gradually. This iterative device selection terminates when a stopping condition is met, i.e., $|\Pi_k|=0$.

For each device in the subset $\Pi_k$, a binary search method is employed for solving Problem \textbf{P3}, and the device $z$ that incurs the smallest increase in the total latency $t_{r,z}^*$ is selected to be included in the set $\mathcal{Q}_k$, as described in Section IV-B(1). This iterative operation is terminated when no device in $\Pi_k$ satisfies \eqref{descent condition}, or the inclusion of any other devices in $\mathcal{Q}_k$ would violate the latency threshold $t_{\mathrm{thr}}$. This proposed scheduling policy is summarized in Algorithm \ref{Alg scheduling}, where \textbf{Theorem 4} is leveraged to ensure the monotonic decrease of \eqref{obj of P4}, until $|\Pi_k|=0$ or constraint \eqref{constraint 1 of P4} is satisfied. A suboptimal solution is then attained.

The binary search method adopted in Alg. \ref{Alg scheduling} is known to incur a computational complexity of $\mathcal{O}\left( \left| \mathcal{Q}_k \right|\log_2 \frac{t_{\mathrm{thr}}}{\epsilon} \right)$ to achieve $\epsilon$-accuracy. 
Since the number of device selection iterations must not exceed $K$ and no more than $K$ binary searches are performed at each device selection iteration, the worst-case complexity of Alg. \ref{Alg scheduling} is $\mathcal{O}\left( K^2\left| \mathcal{Q}_k \right| \log_2 \frac{t_{\mathrm{thr}}}{\epsilon}\right)$, which is no more than $\mathcal{O}\left( K^3 \log_2 \frac{t_{\mathrm{thr}}}{\epsilon}\right)$.
Compared to the exhaustive search with a complexity of $\mathcal{O}\left(2^K \left| \mathcal{Q}_k \right| \log_2 \frac{t_{\mathrm{thr}}}{\epsilon} \right)$, Alg. \ref{Alg scheduling} effectively reduces the computational complexity, especially in the presence of a large number of devices.

\begin{algorithm}[!t]
    \caption{Iterative Scheduling Algorithm}
    \label{Alg scheduling}
    \begin{algorithmic}[1]
    \STATE {Initialize $\mathcal{M}_r \gets \emptyset$, $\mathcal{Q}_1 \gets \emptyset$}
    \STATE {Initial determination: $x\gets\arg\min \limits _{i\in \mathcal{K}} \frac{1+\gamma}{\tau_{r,i}}$, $\mathcal{Q}_1 \gets \mathcal{Q}_1 \cup \{x\}$}
    \FOR{$k=1,2,...$}
        \STATE{ $\mathcal{N}_k \gets \mathcal{K} \setminus \mathcal{Q}_k$}\\
        \STATE{For $\mathcal{N}_k$, determine the available subset $\Pi_k$ with \eqref{descent condition}}
        \IF{$\Pi_k=\emptyset$}
            \STATE{\textbf{break}}
        \ENDIF
        \FOR{each device $y \in \Pi_k$}
            \STATE{With $\mathcal{Q}_k \cup \{y\}$, calculate $t_{r,y}^*$ with a binary search}
        \ENDFOR
    \IF{$\min \limits_{y} t_{r,y}^* \leq t_{\mathrm{thr}} $}
        \STATE{$z \gets \arg\min \limits_{y \in \Pi_k} t_{y}^*$, $\mathcal{Q}_k \gets \mathcal{Q}_k \cup \{z\}$  }
    \ELSIF{$\min \limits_{y} t_{r,y}^* > t_{\mathrm{thr}} $}
        \STATE{\textbf{break}}
    \ENDIF
    \ENDFOR
    \RETURN {$\mathcal{M}_r \gets \mathcal{Q}_k$}
    \end{algorithmic}
\end{algorithm}

\subsection{Linear Programming Solution with Special Structures}
According to \eqref{P2}, as $\gamma = \frac{\eta_{\rm g}}{L} \to 0$, especially when the Lipschitz gradient constant $L$ is large due to a deeper neural network structure and $\eta_{\rm g}$ is relatively small, 
\eqref{obj of P2} is readily approximated by $\frac{1}{M_r}\sum_{i \in \mathcal{M}_r}\frac{1}{\tau_{r,i}}$ since $ \gamma \ll M_r $.
In this case, we define $\boldsymbol{\nu} = \left[ \frac{1}{\tau_{r,1}},\cdots,\frac{1}{\tau_{r,K}} \right]^{\top}$, $\theta = \frac{1}{\boldsymbol{1}^{\top}\boldsymbol{a}}$, $\boldsymbol{\alpha} = {\theta}\boldsymbol{a}$, and $\boldsymbol{\beta} = {\theta}\boldsymbol{b}$. 
Problem \textbf{P2} can be rewritten in a linear form:
\begin{subequations} 
\begin{align}
    \textbf{P5}:\quad &\min_{\boldsymbol{\alpha},\boldsymbol{\beta},{\theta}} \quad  {{\boldsymbol{\nu}^{\top}\boldsymbol{\alpha}}} \label{obj of P5}\\
    \text{s.t.}
            & \quad \alpha_{r,i} c_{r,i} \leq \beta_{r,i} \leq \alpha_{r,i} B  , \label{constraint 0__}\\
		& \sum_{i\in \mathcal{K}}  \beta_{r,i} \leq \theta B ,\quad \sum_{i\in \mathcal{K}} \alpha_{r,i}=1  , \label{constraint 1__}\\
		& 0\leq \alpha_{r,i} \leq 1, 0 < \theta \leq 1 ,\forall i \in \mathcal{K}, \,\,\, \forall r = 1,\cdots,R, \label{constraint 2__}
\end{align} \label{P3}
\end{subequations}
Here the term $c_{r,i}$ shares the same form as $b_{r,i}^*$ in \eqref{optimal bandwidth}, 
with $t_r^* = t_{\mathrm{thr}}$ substituted. To this end, the minimum bandwidth is allocated when any device $i$ is selected.

It is observed that Problem \textbf{P5} is, a convex linear programming problem.
This convex linear program can be effectively solved using, e.g., the Matlab linprog function. Once the optimal $\boldsymbol{\alpha}^*$, $\boldsymbol{\beta}^*$ and $\theta^*$ are obtained, the optima of $\boldsymbol{a}^*$ and $\boldsymbol{b}^*$ can be approximately identified as $\boldsymbol{a}^* = \mathrm{round}\left(\frac{\boldsymbol{\alpha}^*}{\theta^*}\right)$ and $\boldsymbol{b}^* = \mathrm{round}\left(\frac{\boldsymbol{\beta}^*}{\theta^*}\right)$, where $\mathrm{round}(\cdot)$ stands for rounding. 
This linear programming solution complements Alg. \ref{Alg scheduling} to achieve a fast approximate solution with polynomial complexity when $\gamma$ is small.

\subsection{Overall Policy}

\begin{algorithm}[!t]
    \caption{The Proposed Scheduling Algorithm for WFL}
    \label{WFL}
    \begin{algorithmic}[1]
    \STATE {Initialize ${\bf{w}}_1$, $t_{\mathrm{thr}}$, $\eta_{\rm g}$, $\eta_{\mathrm{l}}$ and $\gamma$}
    \FOR{$r=1,2,\dots R$}
        \STATE {Each device $i \in \mathcal{K}$ sends $f_{r,i}$, $\tau_{r,i}$, $p_{r,i}$, $h_{r,i}$ to the BS}
        \STATE {The BS determines the value of $\overline{\tau}_r$.}
        \STATE {The BS derives the scheduled subset $\mathcal{M}_r$ with Alg. \ref{Alg scheduling} or linear programming ($\gamma \to 0$) in Section IV-C.} 
        \STATE{The BS broadcasts the global model ${\bf{w}}_r$, $\overline{\tau}_r$ and the allocated bandwidth $b_{r,i}^*$ to the selected devices.}
        \FOR{each scheduled device $i \in \mathcal{M}_r$ in parallel}
            \STATE {Make the learning rate adjustment $\widetilde{\eta}_{\mathrm{l}} = \frac{\overline{\tau}_r}{\tau_{r,i}}\eta_{\mathrm{l}}$ with the received $\overline{\tau}_r$ }
            \FOR{$j=0,\dots \tau_{r,i}-1$}
                \STATE {Perform local SGD update according to \eqref{adjusted update}}
            \ENDFOR
            \STATE {Send $\widetilde{\Delta}_{r,i} = -\sum_{j=0}^{\tau_{r,i}-1} \widetilde{\eta}_{\mathrm{l}} g_{r,i}^{j}$ back to the BS}
        \ENDFOR
        \STATE {BS receives $\widetilde{\Delta}_{r}=\frac{1}{M}\sum_{i \in \mathcal{M}_r}{\widetilde{\Delta}_{r,i}}$ and update the global model with ${\bf{w}}_{r+1} = {\bf{w}}_{r} + \eta_{\rm g}\widetilde{\Delta}_{r}$}
    \ENDFOR
    \RETURN{${{\bf{w}}_{R+1}}$}
    \end{algorithmic}
\end{algorithm} 

We describe the proposed FLARE framework with the joint device selection and bandwidth allocation strategy for WFL, as outlined in Algorithm \ref{WFL}. Specifically, each round of FL training starts with all devices sending their local computation and channel information to the BS.
After the BS determines the value of $\overline{\tau}_r$, it calls Alg. \ref{Alg scheduling} or its linear programming-based low-complexity alternative to decide the selected device subset $\mathcal{M}_r$, and sends the global model to the selected devices in $\mathcal{M}_r$. Then, the selected devices adjust their local learning rate $\eta_{\mathrm{l}}$ based on $\overline{\tau}_r$. After local computation, with the allocated bandwidths, the selected devices transfer their updated accumulated gradients to the BS for global model aggregation with the global learning rate $\eta_{\rm g}$.

\section{Experiment Results}
In this section, we empirically evaluate the proposed FLARE framework and scheduling strategy for WFL. 

\subsection{Experiment Setup}
We carry out experiments on the MNIST~\cite{MNIST} and CIFAR-10 \cite{cifar} datasets. 
In the case of MNIST, we employ a CNN for handwritten digit recognition. In the case of CIFAR-10, we use another CNN for image classification.
We consider both i.i.d. and non-i.i.d. data distributions for both datasets. In the non-i.i.d. case, the training data samples are sorted by label and distributed randomly among the devices.
Unless specified otherwise, we set the global learning rate \(\eta_{\rm g}=1\), and the local learning rate is set to $\eta_{\mathrm{l}} =0.005$ for MNIST and $\eta_{\mathrm{l}} =0.01$ for CIFAR-10. 
The batch size is $D=40$.

Consider a WFL edge network with $K=40$ devices uniformly distributed between $100$ m and $500$ m away from the BS. 
The path-loss exponent is 3.76. The total bandwidth $B$ is 10 MHz. 
The maximum transmit power of a device is $p_{i,\mathrm{max}}=20$ dBm. 
The power spectrum density of the AWGN is $N_0 = -114$ dBm/MHz at the BS. 
The model size $S$ is $1 \times 10^7$ or $6.4 \times 10^7$ bits for MNIST and CIFAR-10, respectively.
The CPU frequency of a device follows a uniform distribution between 2 GHz and 4 GHz. 
According to our experimental tests, {\color{black}the required number of CPU cycles for computing a data sample is specified empirically to be $110$ cycles/bit for MNIST and $85$ cycles/bit for CIFAR-10.}
To simulate device heterogeneity, we assume that the number of local updates $\tau_{r,i}$ of each device $i$ follows an exponential distribution with the mean value $\tau$ \cite{reisizadeh2020fedpaq}. 
All experiments are conducted on a computer with
an Intel 13700K processor and 64 GB memory, running Python 3.9, Numpy 1.23.3, and PyTorch 1.13.0, installed on a Windows 10 operating system.

\subsection{Evaluation of FLARE}

Fig. 3 demonstrates the benefit of the local learning rate adjustments in FLARE in the presence of device and data heterogeneity, where we assume that the wireless bandwidth is sufficient.
The following baselines are considered. 

\begin{figure}[t!]
	\centering
	\includegraphics[width=0.8\linewidth]{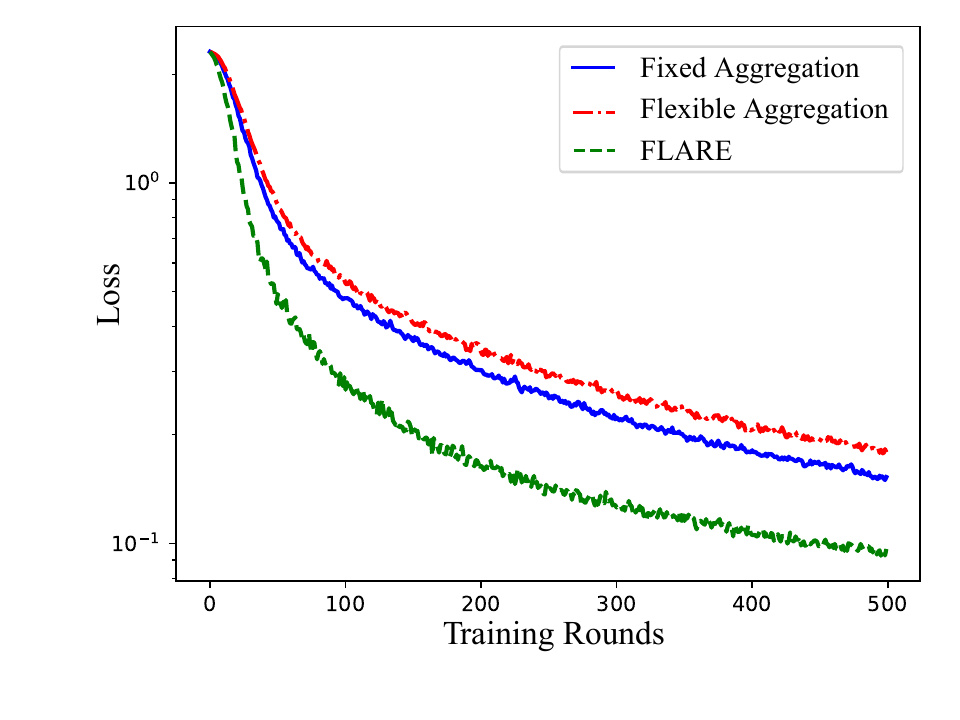}
	\caption{Comparison of training loss on non-i.i.d MNIST dataset with $K=60$, $M_r=20$, and uniform sampling. For the fixed aggregation, we have $\tau_{r,i}=7, \forall r \in \mathcal{M}_r$. For the flexible aggregation, every 20 devices among the $K$ devices perform 12, 6, and 3 local updates, respectively. For FLARE, $\bar{\tau}_r$ is set as the maximal local updates among the selected devices in each round.}
	\label{fig:loss}
\end{figure}

\begin{itemize}
    
    \item 
    Fixed Aggregation: The same number of local training iterations with persistent learning rates is performed at the selected devices per round.

    \item Flexible Aggregation: Different numbers of local training iterations with persistent learning rates can be performed at the selected devices per round.

    
\end{itemize}
We see that FLARE dramatically outperforms the other two settings, attributed to the adaptively adjusted local learning rates of the selected devices. 

Under FLARE, we further consider the following strategies that the BS can take to determine $\overline{\tau}_r$ in each round $r$.
\begin{itemize}
    \item Max strategy (MaS): $\overline{\tau}_r = \mathop {\max }\limits_{i\in \mathcal{M}_r} \{ {\tau _{r,i}}\} $ for $ r\geq 1$;
    \item Mean strategy (MeS): $\overline{\tau}_r = \mathop {\mathrm{mean} }\limits_{i\in \mathcal{M}_r} \{ \tau _{r,i} \} $ for $ r\geq 1$;
    \item Fixed max strategy (FMaS): $\overline{\tau}_r = \mathop {\max }\limits_{i\in \mathcal{M}_r} \{ {\tau _{1,i}}\} $ for $ r>1$;
    \item Fixed mean strategy (FMeS): $\overline{\tau}_r = \mathop {\mathrm{mean} }\limits_{i\in \mathcal{M}_r} \{ \tau _{1,i} \} $ for $ r>1$.   
\end{itemize}
Here, FMaS and FMeS suggest that the BS directly utilizes $\{\bar{\tau}_{1,i}\}$ from the first training round in various ways.

Fig. \ref{fig_tau} shows the convergence of FLARE with $\tau=3$. A random, uniform selection of 10 out of $K=40$ devices and a non-uniform device selection are considered.
As for the non-uniform device selection, we divide all devices into four groups before training. The probabilities of devices being selected in the four groups are $0.05, 0.15, 0.2$, and $0.6$ per round.
MaS and MeS outperform FedAvg upon convergence, by up to 4.5\% and 9.4\% under the i.i.d. and non-i.i.d. settings, respectively. 
MaS converges faster than MeS because $\bar{\tau}_r$ is larger in most rounds and hence the test accuracy is better under MaS.
In contrast, MeS allows for a stable accuracy growth, as it reduces $\phi_1$ and $\phi_2$ in \eqref{one round error}. 
Moreover, FLARE tolerates non-uniform device selection with limited accuracy loss, e.g., only 2.4\% under MaS, as revealed by comparing Figs. \ref{fig_tau}(b) and \ref{fig_tau}(d).

Under FMaS and FMeS, $\overline{\tau}_r$ is adjusted. FMeS tends to offer greater robustness since no excessively large $\overline{\tau}_r$ is produced.
However, the gain of FMeS is marginal, compared to FedAvg. Additionally, FMaS and FMeS only introduce negligible additional communications between the BS and devices in the first training round, making them slightly more communication-efficient than MaS and MeS.

\begin{figure}[t!]
	\centering
	\subfloat[{\centering Accuracy on i.i.d. MNIST with uniform sampling.}]{\includegraphics[width=1.71in]{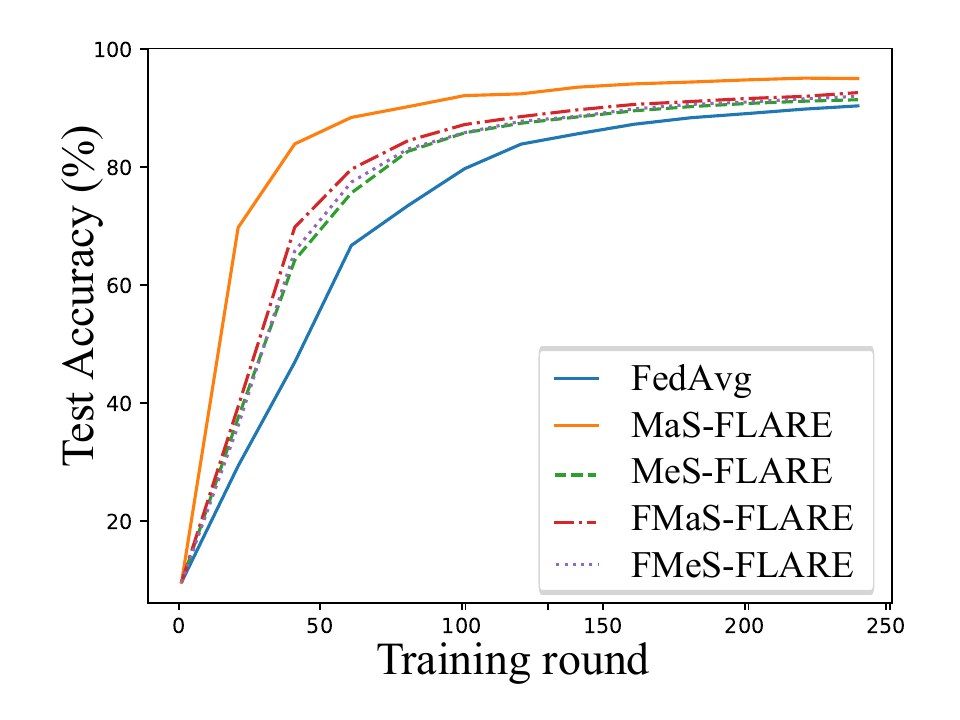}
		\label{fig_iid_case}} 
	\subfloat[{\centering Accuracy on non-i.i.d. MNIST with uniform sampling.}]{\includegraphics[width=1.72in]{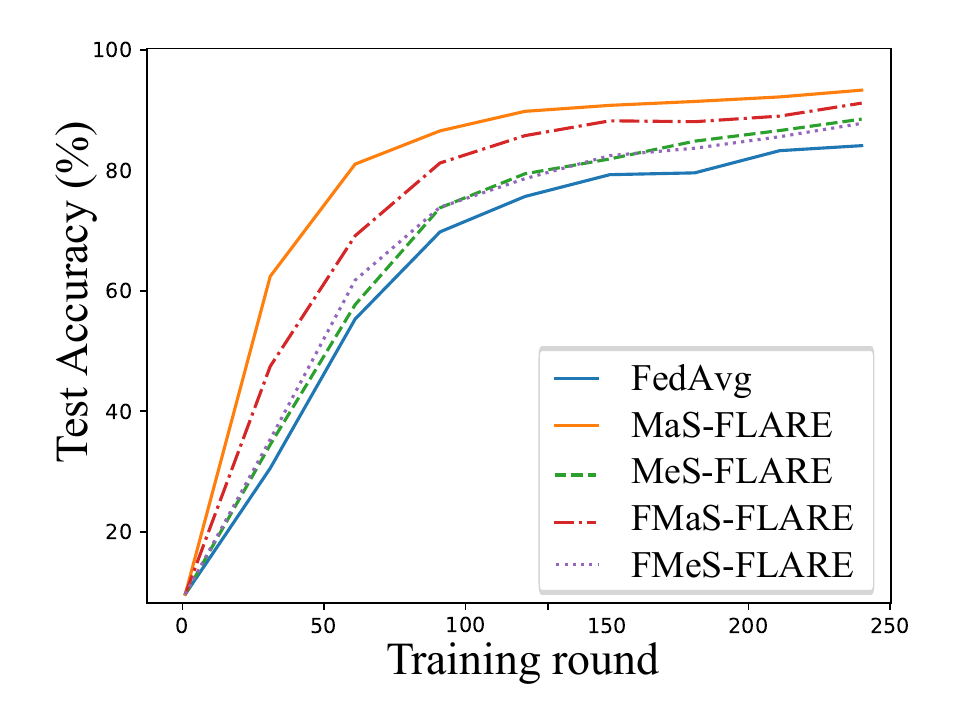}
		\label{fig_non_iid_case}}
        \hfil
        \subfloat[{\centering Accuracy on i.i.d. MNIST with non-uniform sampling.}]{\includegraphics[width=1.71in]{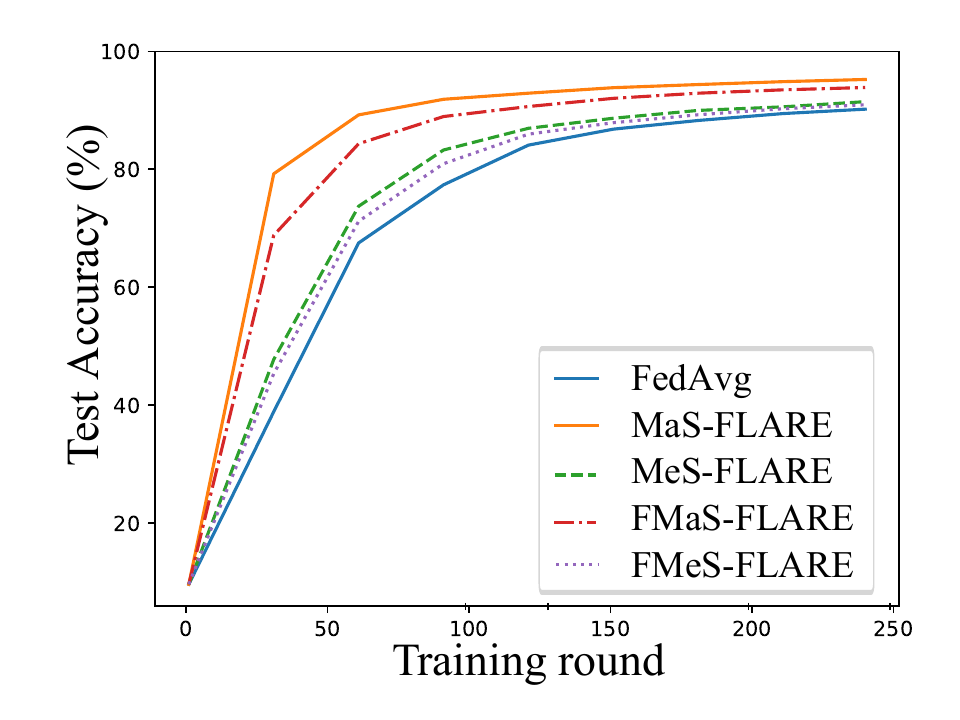}
		\label{fig_iid_case_non_uni}}
        \subfloat[{\centering Accuracy on non-i.i.d. MNIST with non-uniform sampling.}]{\includegraphics[width=1.71in]{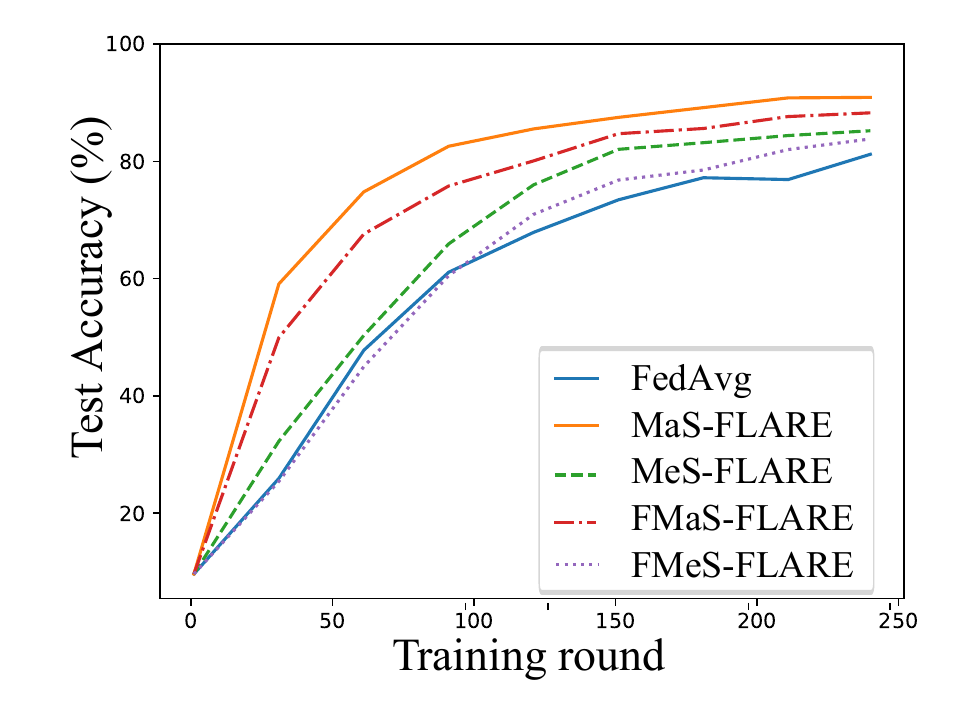}
		\label{fig_non_iid_case_non_uni}}
	\caption{The performance of different $\bar{\tau}_r$-strategies of FLARE on MNIST with both uniform and non-uniform sampling.}
	\label{fig_tau}
\end{figure}

\subsection{Comparison with State-Of-The-Art}
We proceed to evaluate the FLARE framework for resource-constrained WFL, by considering the following benchmarks.

\begin{itemize}
    \item Pre-tuned Scheduling (PS) \cite{9207871}: The BS randomly selects $M_r$ devices in each round $r$. The optimal value of $M_r$ is pre-tuned experimentally.

    \item Channel Proportion (CP) \cite{8851249}: The BS selects the device subset $\mathcal{M}_r$ with the best instantaneous channel qualities iteratively until the delay constraint \eqref{constraint 2} is satisfied;
    that is, based on the channel states of the devices, the BS selects devices one after another and utilizes \textbf{Theorem~3} to allocate the available bandwidth until $t_{\mathrm{thr}}$ is reached.
    
    \item Device Maximization (DM) \cite{8761315}: The BS iteratively selects the devices that minimize the growth of the total delay $t_r$, and then distributes the available bandwidth $B$ evenly among the selected devices.
    \item Computation Minimization (CM) \cite{9139810}: The BS ranks the expected computing times of all devices, followed by a binary search for the largest device subset satisfying all constraints. 
\end{itemize}
For fair comparisons,  the optimal bandwidth allocation policy in \textbf{Theorem 3} is also applied in PS, CP, and CM. We consider two ways to implement the proposed scheduling policy under FLARE. The one utilizing the MaS strategy to determine $\overline{\tau}_r$ round-by-round is dubbed SPF. The one without the learning rate adjustments is called SP.

\begin{figure}[!t]
\centering
\subfloat[{\centering Accuracy on i.i.d. MNIST with $t_{\mathrm{thr}} = 0.4$.}]{\includegraphics[width=1.7in]{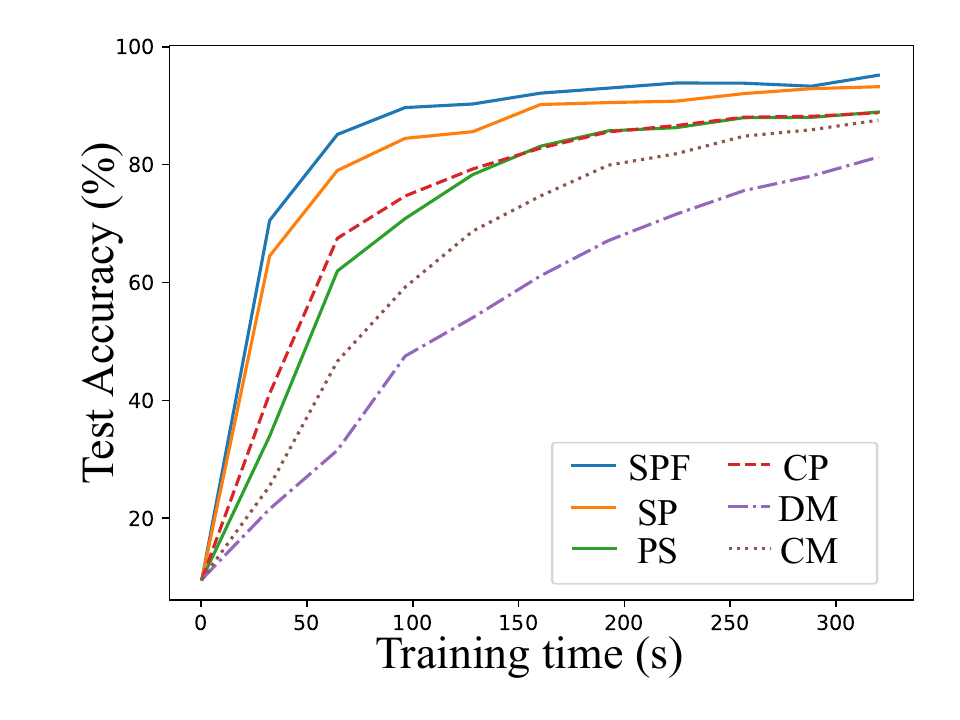}%
\label{fig_5_1}}
\subfloat[{\centering Accuracy on non-i.i.d. MNIST with $t_{\mathrm{thr}} = 0.4$.}]{\includegraphics[width=1.7in]{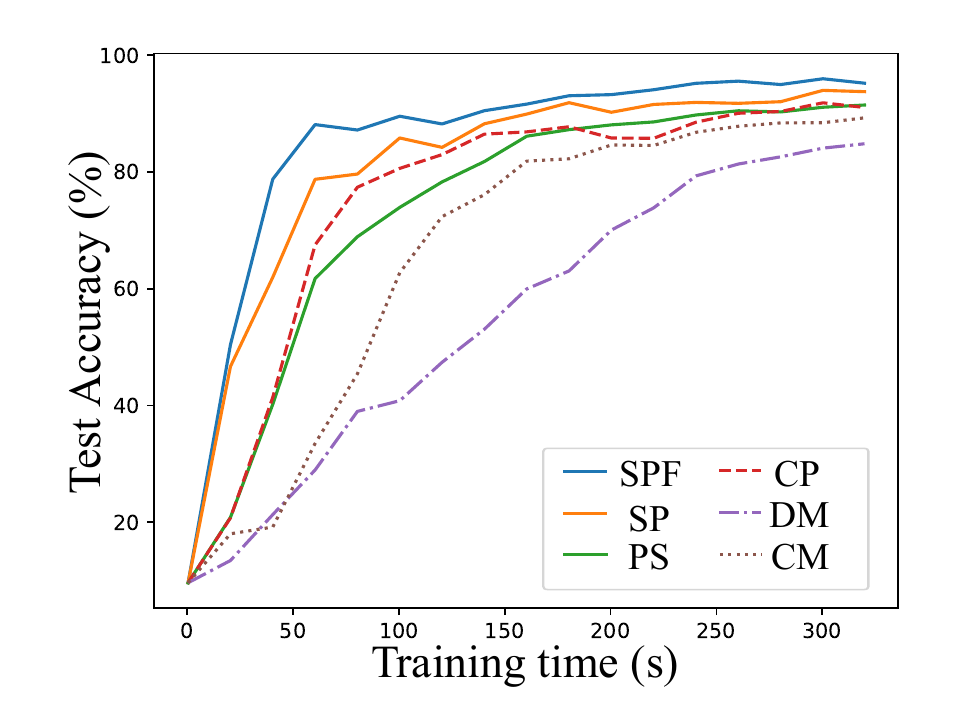}%
\label{fig_5_2}}
\hfil
\subfloat[{\centering Accuracy on i.i.d. MNIST with $t_{\mathrm{thr}} = 1$.}]{\includegraphics[width=1.7in]{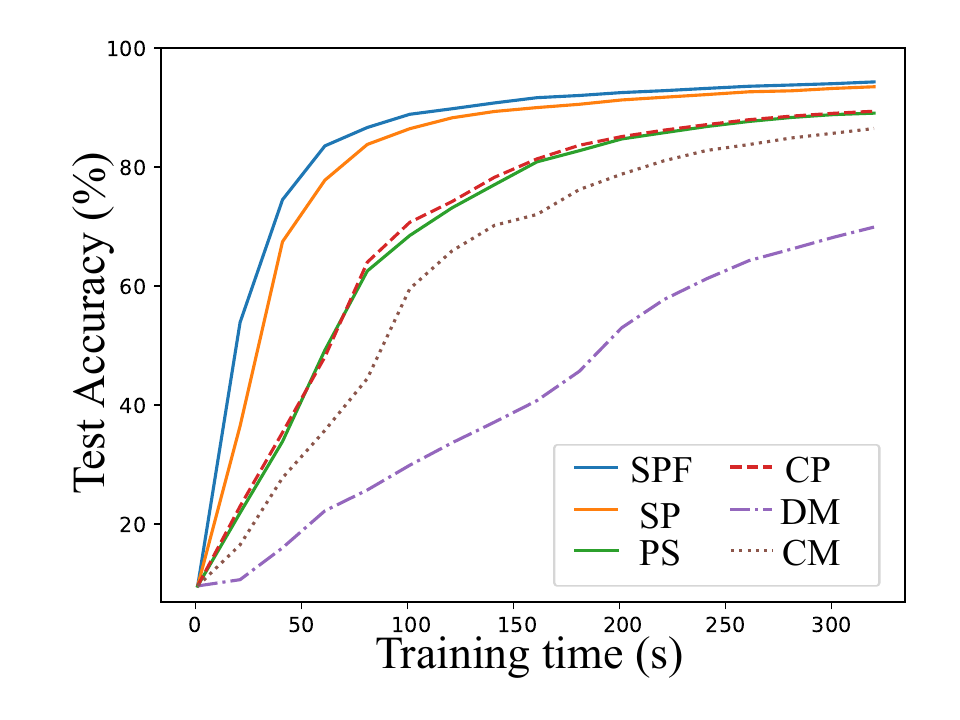}%
\label{fig_5_3}}
\subfloat[{\centering Accuracy on non-i.i.d. MNIST with $t_{\mathrm{thr}} = 1$.}]{\includegraphics[width=1.68in]{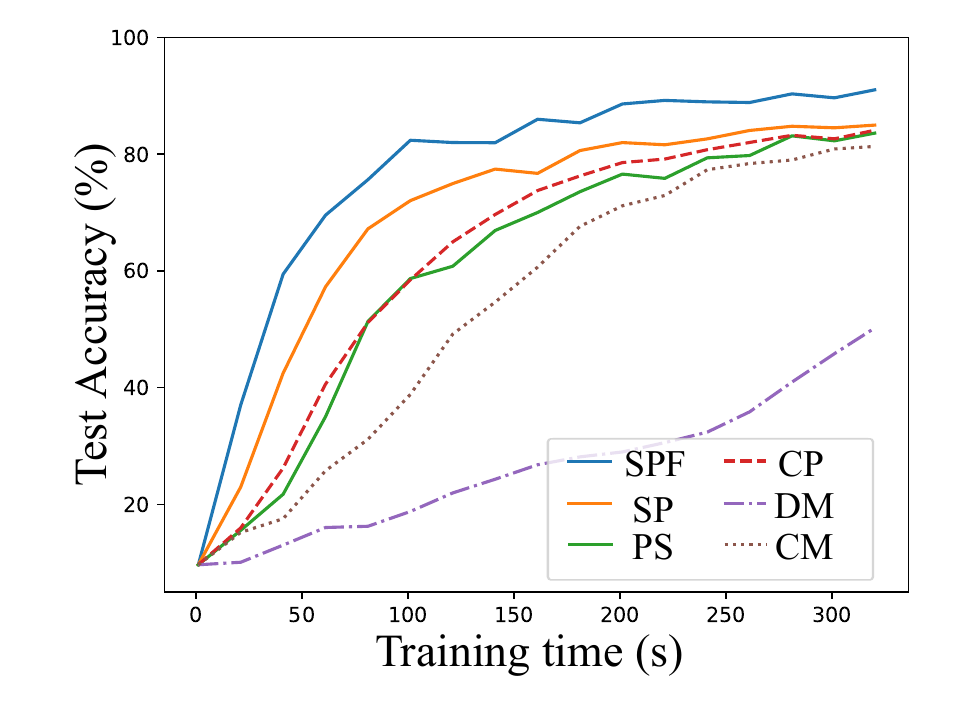}%
\label{fig_5_4}}
\caption{Convergence performance of different policies on MNIST.}
\label{fig_5}
\end{figure}

The convergence performances of the different policies and datasets are plotted in Figs. 5 and 6. For MNIST, two settings are considered: $t_{\mathrm{thr}}=0.4$ and $1$. For CIFAR-10, $t_{\mathrm{thr}}=1.5$ and $5$.
In both figures, SPF and SP consistently outperform the benchmarks under both i.i.d. and non-i.i.d. datasets.
CM performs the worst, followed by DM in all cases. This is due to their tendency to favor devices with insufficient local updates, thus hindering convergence,
even though DM considers channel states.
Similarly, CP focuses solely on the channel strengths of the devices and 
overlooks their different computational abilities and unequal contributions to model updates,
while PS lacks adaptability to time-varying environments due to a persistent number of devices selected per round.

By comparing SPF and SP, it is evident that the inclusion of FLARE consistently improves performance, especially under the non-i.i.d. settings. This can be attributed to the exacerbation of differences between the local models induced by heterogeneous updates in non-i.i.d. distributions, leading to an increased performance gap between SP and SPF.

\begin{figure}[!t]
\centering
\subfloat[{\centering Accuracy on i.i.d. CIFAR-10 with $t_{\mathrm{thr}} = 1.5$.}]{\includegraphics[width=1.7in]{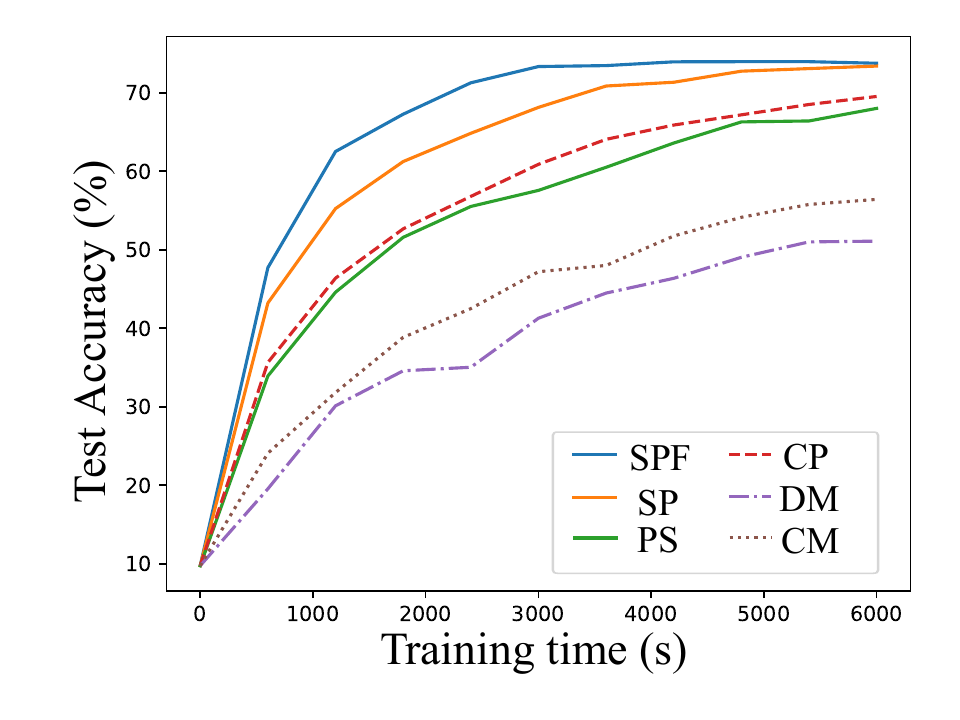}%
\label{fig_6_1}}
\subfloat[{\centering Accuracy on non-i.i.d. CIFAR-10 with $t_{\mathrm{thr}} = 1.5$.}]{\includegraphics[width=1.7in]{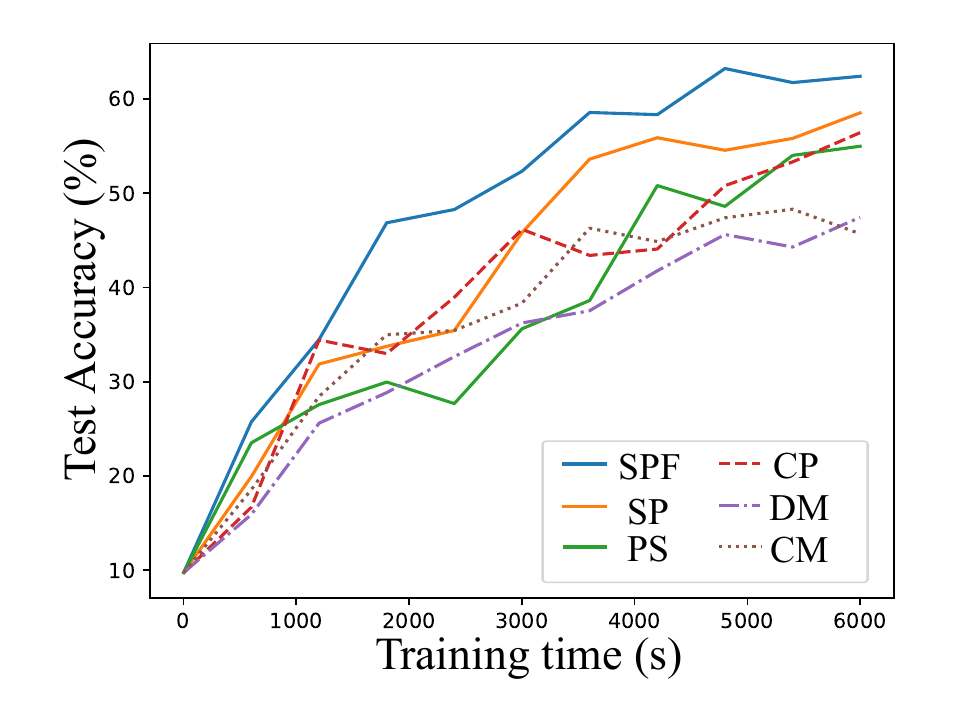}%
\label{fig_6_2}}
\hfil
\subfloat[{\centering Accuracy on i.i.d. CIFAR-10 with $t_{\mathrm{thr}} = 5$.}]{\includegraphics[width=1.7in]{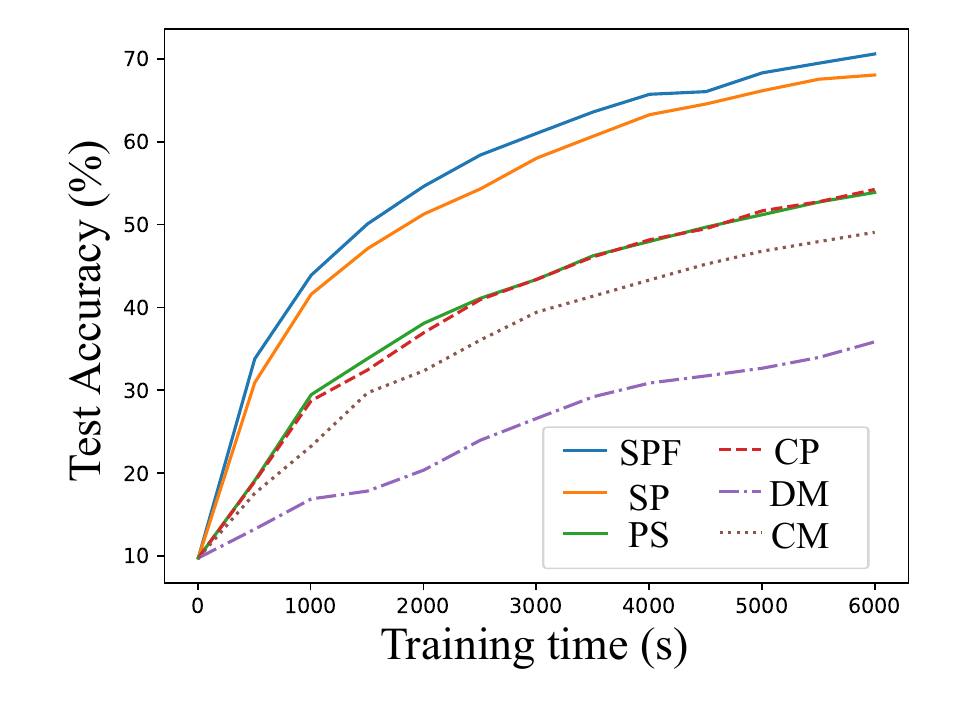}%
\label{fig_6_3}}
\hfil
\subfloat[{\centering Accuracy on non-i.i.d. CIFAR-10 with $t_{\mathrm{thr}} = 5$.}]{\includegraphics[width=1.7in]{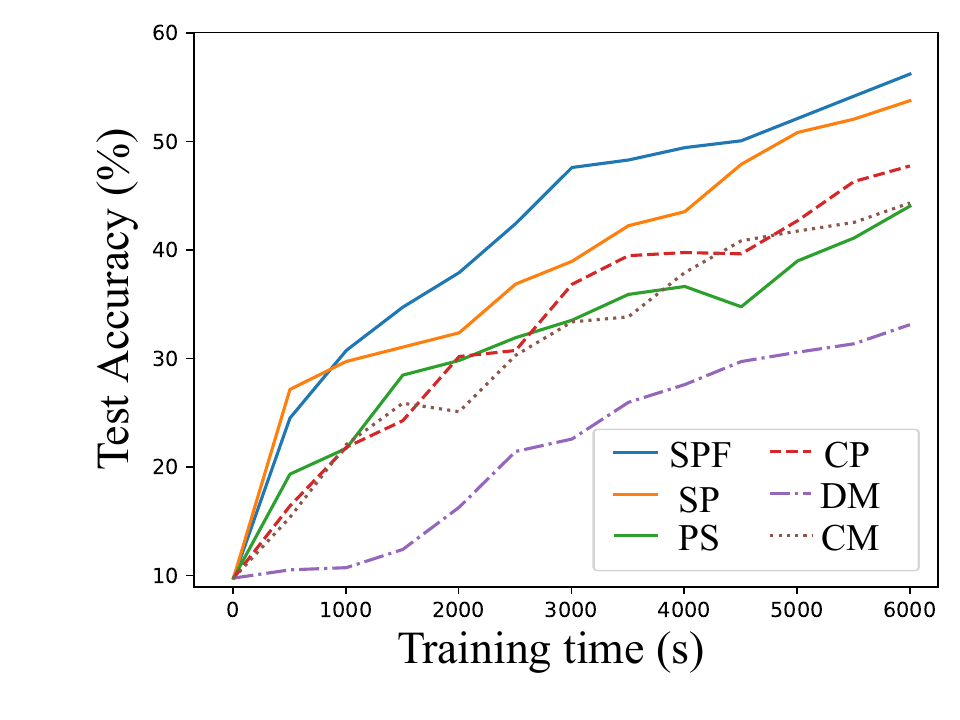}%
\label{fig_6_4}}
\caption{Convergence performance of different policies on CIFAR-10. }
\label{fig_6}
\end{figure}

We examine the average number of participating devices per round and their local update numbers under different policies in Fig. 7. It is noticed that SPF/SP demonstrates consistently the highest number of local updates, 
and maintains a sufficient number of participating devices
in both the cases depicted in Figs. 7(a) and 7(b). This indicates that it is crucial to select devices with a larger number of local updates (and hence higher contributions) while ensuring adequate device participation,  to accelerate the convergence.

\begin{figure}[!t]
	\centering
	\subfloat[non-i.i.d. MNIST.]{\includegraphics[width=1.698in]{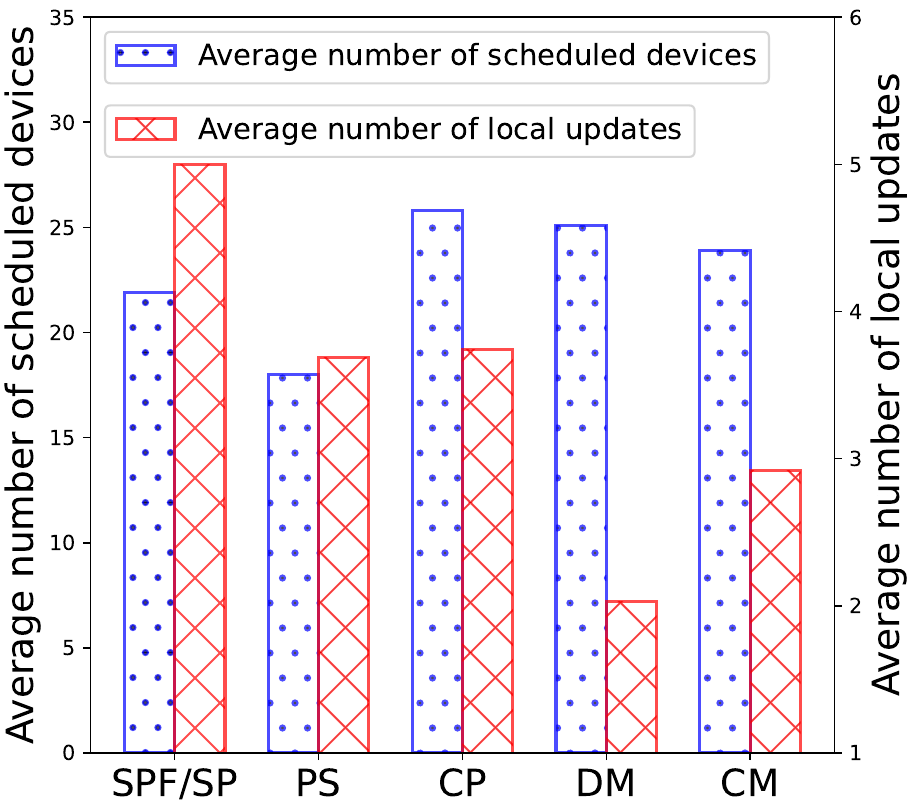}
		\label{fig_7_1}}
	\subfloat[non-i.i.d. CIFAR-10.]{\includegraphics[width=1.698in]{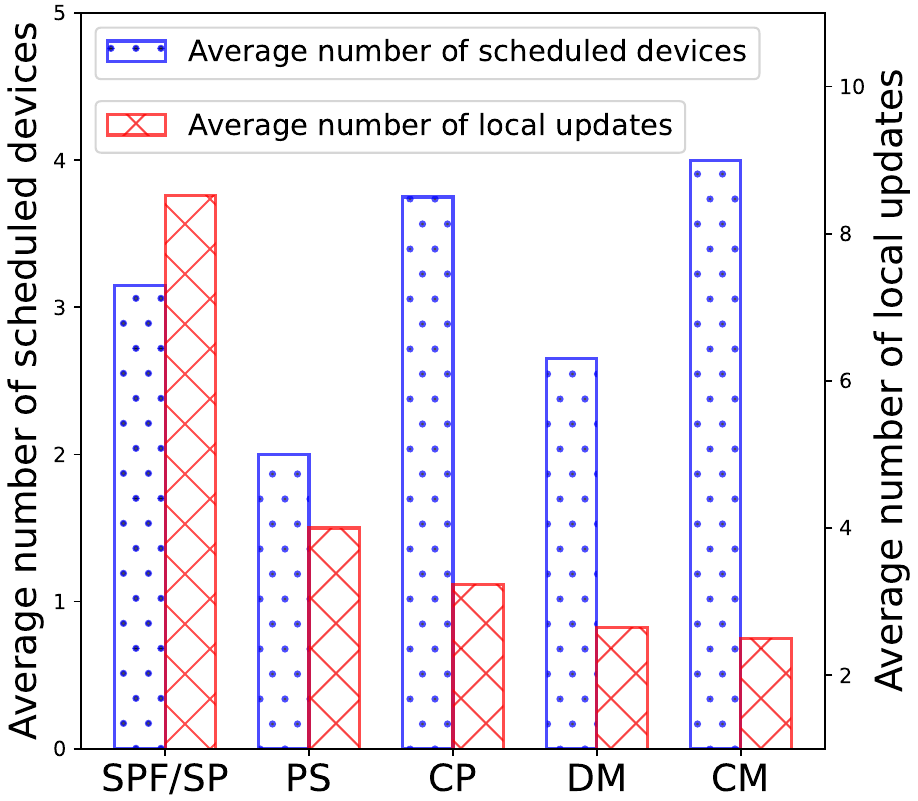}
		\label{fig_7_2}}
	\caption{Average participated devices per round and local updates per device of different policies with $t_{\mathrm{thr}}=2$. For SPF/SP, the participated device subset is determined by calling Alg. 1 at the BS side in each round.}
	\label{fig_7}
\end{figure}

\begin{figure}[!t]
	\centering
	\includegraphics[width=0.82\linewidth]{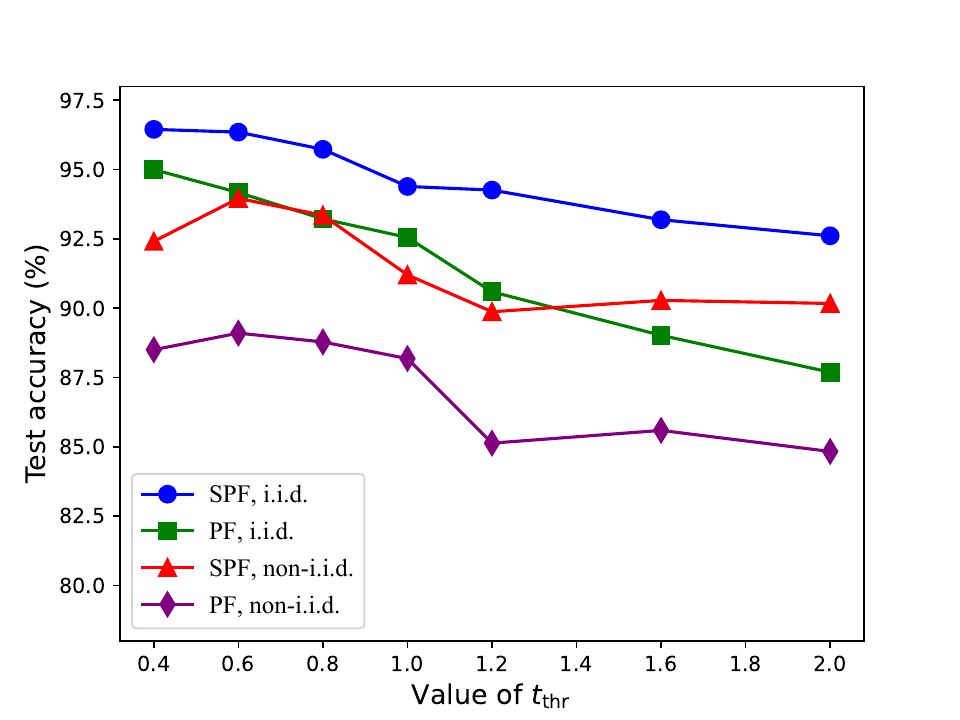}
	\caption{The maximum achievable accuracy on i.i.d. and non-i.i.d. MNIST dataset. The total training time $t$ is 300s.}
	\label{fig_8}
\end{figure} 

Last but not least, we investigate the impact of $t_{\mathrm{thr}}$ on the achievable accuracy of SPF and SP within a fixed training time. Specifically, we explore the values of $t_{\mathrm{thr}}$ from $\{ 0.4, 0.6, 0.8, 1, 1.2, 1.6, 2\}$. 
As $t_{\mathrm{thr}}$ increases, the system becomes more tolerant of larger $t_{r,i}^{\mathrm{comp}}$ and $t_{r,i}^{\mathrm{comm}}$, resulting in an increased number of devices selected according to Alg. \ref{Alg scheduling}. 
Meanwhile, the number of training rounds decreases.
As shown in Fig. 8, in the i.i.d. case where data heterogeneity does not exist, the performance degenerates with $t_{\mathrm{thr}}$ since the enlarging subset of selected devices cannot fully compensate for the reduction in the number of training rounds under SPF and PF.
This contrasts with the non-i.i.d. case, where the optimal performance is acquired at $t_{\mathrm{thr}}=0.6$, indicating a trade-off between training time and training rounds. 
Consistent accuracy is also attained when $t_{\mathrm{thr}} > 0.6$. This stability stems from the fact that the increasing number of devices can counteract effectively the adverse effect of data heterogeneity.

\section{Conclusion}
In this paper, we presented the FLARE framework specifically designed to address data, device, and channel heterogeneity and expedite WFL training. Our contribution included a comprehensive convergence analysis of FLARE with non-convex loss functions under data and device heterogeneity and arbitrary device scheduling policy.
We optimized the device selection and bandwidth allocation of FLARE in the face of resource constraints by minimizing the convergence upper bound. Efficient solutions were developed by revealing a nested optimization structure of the intended problem, as well as a much simpler linear structure when the neural network models have large Lipschitz constants. 
Experiments corroborated the efficacy of FLARE in mitigating heterogeneity across diverse system settings, and that the proposed scheduling policy consistently outperformed the state of the art in test accuracy.

\section*{Appendix A \\ Proof of Theorem 1}
We begin with the following general assumption:
\begin{equation}\label{App A 1}
\mathbb{E}_{\xi}\left[\|\nabla f(\mathbf{x} ; \xi)-\nabla f(\mathbf{x})\|^2\right] \leq \sigma^2 \text{ and } \|\nabla f(\mathbf{x})\|^2 \leq g.
\end{equation}  
Setting the local learning rate $\widetilde{\eta_{\mathrm{l}}} = \eta_{\mathrm{l}} \bar{\tau}_r / {\tau}_r$ in round $r$ and denoting $\xi^i$ as the local data used for the $i$-th local iteration in an SGD round, we have
\begin{align}
&\left\|\widetilde{\boldsymbol{\rm{w}}} \!-\! \boldsymbol{\rm{w}}\right\|^2  \!=\!  \left\|\widetilde{\eta}_{\mathrm{l}}\sum_{i=0}^{{\tau_r}-1}\nabla f(\widetilde{\mathbf{w}}^i,\xi^i) -\eta_{\mathrm{l}}\sum_{j=0}^{\bar{\tau}_r-1}\nabla f(\mathbf{w}^j,\xi^j)  \right\|^2 \notag \\
&\!\!\!\! \overset{(a)}{\leq} 2 \tau_r\widetilde{\eta}_{\mathrm{l}}^2\sum_{i=0}^{{\tau_r}-1}\left\|\nabla f(\widetilde{\mathbf{w}}^i,\xi^i)-\nabla f(\widetilde{\mathbf{w}}^i) + \nabla f(\widetilde{\mathbf{w}}^i)\right\|^2  \notag \\
&+ 2\bar{\tau}_r{\eta_{\mathrm{l}}^2}\sum_{j=0}^{\bar{\tau}_r-1}\left\|\nabla f(\mathbf{w}^j,\xi^j)-\nabla f(\mathbf{w}^j) + \nabla f(\mathbf{w}^j)\right\|^2 \notag \\
&\!\!\!\!\overset{(b)}{\leq} 4 {\tau}_r\widetilde{\eta}_{\mathrm{l}}^2\sum_{i=0}^{{\tau}_r-1}\left(\left\|\nabla f(\widetilde{\mathbf{w}}^i,\xi^i)\!-\!\nabla f(\widetilde{\mathbf{w}}^i)\right\|^2 \!+\! \left\|\nabla f(\widetilde{\mathbf{w}}^i)\right\|^2 \right) \notag \\
&\!\!\! + \! 4\bar{\tau}_r{\eta_{\mathrm{l}}^2}\sum_{j=0}^{\bar{\tau}_r-1}\left(\left\|\nabla f(\mathbf{w}^j,\xi^j)\!-\!\nabla f(\mathbf{w}^j) \right\|^2 \!+\! \left\|\nabla f(\mathbf{w}^j)\right\|^2 \right)\nonumber \\
&\overset{(c)}{\leq} 2 \cdot 4\bar{\tau}_r^2\eta_{\mathrm{l}}^2(g + \sigma^2) = 8\eta_{\mathrm{l}}^2\bar{\tau}_r^2(g \!+\! \sigma^2). 
\end{align}
where (a) and (b) are obtained because of the Cauchy-Schwartz inequality, and (c) is obtained by substituting \eqref{App A 1}. 

\section*{Appendix B \\ Proof of Lemma 1}

Since $f_i$, $\forall i \in \mathcal{K}$ is $L$-smooth, $F({\bf{w}})$ is $L$-smooth. Since $\nabla F$ is also Lipschitz continuous, it readily leads to
\begin{equation} \label{overall 1}
{\mathbb{E}[F({{\bf{w}}_{r \!+\! 1}})\!-\!F({{\bf{w}}_r})] \leq \underbrace{\mathbb{E}\left[\!\left\langle {\nabla F({{\bf{w}}_r}),{\eta_{\rm g}\Delta _r}} \!\right\rangle \right]}_{A_1} \!+ {\underbrace {\frac{L}{2}\mathbb{E}\left[\!{{\left\| {{\eta_{\rm g}\Delta _r}} \right\|}^2}\!\right]}_{A_2} }},
\end{equation}
where $\mathbb{E}\left[ \cdot \right]$ takes expectation over random minibatch sampling.

The upper bound of $A_1$ is given by
\begin{align} \label{A1_}
&A_1 = -\eta_{\rm g}\eta_{\mathrm{l}}\left\langle {\nabla F({{\bf{w}}_r}),\mathbb{E}\left[\sum_{i \in \mathcal{M}_r}\frac{1}{M_r}\sum_{j=0}^{\tau_{r,i}-1}\frac{\bar{\tau}_r}{\tau_{r,i}}g_{r,i}^j\right] } \right\rangle  \nonumber    \\ 
&= -\eta_{\rm g}\eta_{\mathrm{l}}\left\langle {\sqrt{\bar{\tau}_r} \nabla F({{\bf{w}}_r}),\frac{1}{M_r}\sum_{i \in \mathcal{M}_r}\sum_{j=0}^{\tau_{r,i}-1}\frac{\sqrt{\bar{\tau}_r}}{\tau_{r,i}}\nabla f_i({\bf{w}}_{r,i}^j) } \right\rangle   \nonumber    \\ 
&\overset{(a)}{\leq }\! a_1 \! +\! \frac{\eta_{\rm g}\eta_{\mathrm{l}}}{2}\left\|\frac{1}{M_r}\!\!\!\sum_{i \in \mathcal{M}_r}\!\!\sum_{j=0}^{\tau_{r,i}-1}\!\!\frac{\sqrt{\bar{\tau}_r}}{\tau_{r,i}} \left(\nabla f_i({\bf{w}}_{r,i}^j)\!-\!\nabla f_i({{\bf{w}}_r}) \right) \right\|^2  \nonumber    \\ 
&\overset{(b)}{\leq }  a_1 +  \underbrace{ {\frac{L^2\eta_{\rm g}\eta_{\mathrm{l}}}{2M_r}\sum_{i \in \mathcal{M}_r} \frac{\bar{\tau}_r}{\tau_{r,i}}\sum_{j=0}^{\tau_{r,i}-1}\mathbb{E}\left[\left\| {\bf{w}}_{r,i}^j-{{\bf{w}}_r}\right\|^2\right]} }_{B_1},
\end{align}
where $a_1=-\frac{\bar{\tau}_r\eta_{\rm g}\eta_{\mathrm{l}}}{2}\left\|\nabla F({{\bf{w}}_r}) \right\|^2 $. Here, (a) is due to $-\!2\left\langle {\bf{a}},{\bf{b}} \right\rangle \!\leq\! -\| {\bf{a}} \|^2\! + \!\| {\bf{a}}\!-\!{\bf{b}} \|^2$ with ${\bf{a}}\! =\! \sqrt{\bar{\tau}_r} \nabla F({{\bf{w}}_r})$ and ${\bf{b}} \!=\!\frac{1}{M_r}\sum_{i \in \mathcal{M}_r}\sum_{j=0}^{\tau_{r,i}\!-\!1}\frac{\sqrt{\bar{\tau}_r}}{\tau_{r,i}}\nabla f_i({\bf{w}}_{r,i}^j) $; and (b) is from \textbf{Assumption~1}.

To derive the upper bound of $B_1$, we start with the upper bound of $\mathbb{E}\left[ \left\|{\bf{w}}_{r,i}^j-{{\bf{w}}_r} \right\|^2 \right] \forall i,j $. Note that $\mathbb{E}\left[ \left\|{\bf{w}}_{r,i}^j-{{\bf{w}}_r} \right\|^2 \right] = 0$
when $j=0$, since $ \boldsymbol{\rm{w}}_{r,i}^0=\boldsymbol{\rm{w}}_r $ by definition. When $j\geq1$, we have
\begin{align}\label{single bias 1}
   & \mathbb{E}\left[ \left\|{\bf{w}}_{r,i}^j-{{\bf{w}}_r} \right\|^2 \right]  = \mathbb{E}\left[ \left\|{\bf{w}}_{r,i}^{j-1}-\frac{\bar{\tau}_r}{\tau_{r,i}}\eta_{\mathrm{l}}{g}_{r,i}^{j-1}-{\bf{w}}_{r} \right\|^2 \right]  \nonumber    \\
    &\overset{(a)}{=} \mathbb{E}\left[ \left\|{\bf{w}}_{r,i}^{j-1}-\frac{\bar{\tau}_r}{\tau_{r,i}}\eta_{\mathrm{l}}\nabla {f}_i({\bf{{w}}}_{r,i}^{j-1})-{\bf{w}}_r \right\|^2 \right] + \frac{\bar{\tau}_r^2}{\tau_{r,i}^2}\sigma^2\eta_{\mathrm{l}}^2  \nonumber    \\ 
    &\overset{(b)}{\leq} \left(1+\frac{1}{b\tau_{r,i}-1}\right)\mathbb{E}\left[ \left\| {\bf{w}}_{r,i}^{j-1}-{\bf{w}}_r \right\|^2 \right] \nonumber    \\  
    &+ (1+b\tau_{r,i}-1)\frac{\bar{\tau}_r^2}{\tau_{r,i}^2}\eta_{\mathrm{l}}^2\left\|\nabla f_i({\bf{w}}_{r,i}^{j-1})\right\|^2 + \frac{\bar{\tau}_r^2}{\tau_{r,i}^2}\sigma^2\eta_{\mathrm{l}}^2  \nonumber    \\ 
    &\overset{(c)}{\leq} \left(1+\frac{1}{b\tau_{r,i}-1}\right)\mathbb{E}\left[ \left\| {\bf{w}}_{r,i}^{j-1}-{\bf{w}}_r \right\|^2 \right]+ \frac{\bar{\tau}_r^2}{\tau_{r,i}^2}\sigma^2\eta_{\mathrm{l}}^2  \nonumber    \\  
    &\!+\! \frac{2b\bar{\tau}_r^2\eta_{\mathrm{l}}^2}{\tau_{r,i}}\left\|\nabla f_i({\bf{w}}_{r,i}^{j-1})\!-\!\nabla f_i({\bf{w}}_r)\right\|^2 \!+\! \frac{2b\bar{\tau}_r^2\eta_{\mathrm{l}}^2}{\tau_{r,i}}\left\|\nabla f_i({\bf{w}}_r)\right\|^2\nonumber    \\
    &\overset{(d)}{\leq} \left( 1+\frac{1}{b\tau_{r,i}-1}+2b\frac{L^2\bar{\tau}_r^2}{\tau_{r,i}}\eta_{\mathrm{l}}^2 \right)\mathbb{E}\left[ \left\| {\bf{w}}_{r,i}^{j-1}-{\bf{w}}_r \right\|^2 \right] \nonumber    \\
    &+ \frac{2b\bar{\tau}_r^2\eta_{\mathrm{l}}^2}{\tau_{r,i}}\left\|\nabla f_i({\bf{w}}_r)\right\|^2 + \frac{\bar{\tau}_r^2}{\tau_{r,i}^2}\sigma^2\eta_{\mathrm{l}}^2  \nonumber    \\
    &\!\!\overset{(e)}{\leq}\! \frac{b\tau_{r,i} \!+\! 1}{b\tau_{r,i} \!-\! 1}\mathbb{E}\!\!\left[ \left\| {\bf{w}}_{r,i}^{j-1} \!\!-\!\! {\bf{w}}_r \right\|^2 \right] \! \!+\! \frac{\bar{\tau}_r^2\sigma^2\eta_{\mathrm{l}}^2}{\tau_{r,i}^2} \!+\!  \frac{2b\bar{\tau}_r^2\eta_{\mathrm{l}}^2}{\tau_{r,i}}\left\|\nabla f_i({\bf{w}}_r)\right\|^2,
    \end{align}
where (a) holds since $\mathbb{E}\left[\|\mathbf{x}\|^2\right]=\mathbb{E}\left[\|\mathbf{x}-\mathbb{E}[\mathbf{x}]\|^2\right]+\|\mathbb{E}[\mathbf{x}]\|^2$; (b) is due to the relaxed triangle inequality $\left\|{\bf{x}}_1+{\bf{x}}_2\right\|^2 \leq(1+\beta)\left\|{\bf{x}}_1\right\|^2+\left(1+\frac{1}{\beta}\right)\left\|{\bf{x}}_2\right\|^2, \forall \beta = b\tau_{r,i}-1 \geq 0$; 
(c) comes from the Cauchy-Schwartz inequality; (d) is based on \textbf{Assumption 1}; (e) holds under the condition that $2b\frac{L^2\bar{\tau}_r^2}{\tau_{r,i}}\eta_{\mathrm{l}}^2 \leq \frac{1}{b\tau_{r,i}-1}$. 
Without loss of generality, we set $b=1.2$  to ensure $\beta = b\tau_{r,i}-1>0$ since $\tau_{r,i} \geq 1$. 
Moreover, $\eta_{\mathrm{l}} \leq \frac{0.58}{L\bar{\tau}_r} $ is required by resolving the condition $2.4\frac{L^2\bar{\tau}_r^2}{\tau_{r,i}}\eta_{\mathrm{l}}^2 \leq \frac{1}{1.2\tau_{r,i}-1}$.

By unrolling the recursion of \eqref{single bias 1}, 
$B_1$ is upper bounded by
\begin{align}    \label{B1_}
&B_1 \leq \mathbb{E}\left[ \frac{L^2\eta_{\rm g}\eta_{\mathrm{l}}}{2 M_r}\sum_{i \in \mathcal{M}_r}\frac{\bar{\tau}_r e_{r,i}}{\tau_{r,i}} \right]\nonumber \\
&\overset{(a)}{\leq} \frac{L^2 \eta_{\rm g}\eta_{\mathrm{l}}}{2 M_r} \mathbb{E}\left[\sum_{i \in \mathcal{M}_r} 2.4 \bar{\tau}_r^3 \eta_{\mathrm{l}}^2\left\|\nabla f_i\left({\bf{w}}_r\right)\right\|^2+\frac{\bar{\tau}_r^3 \eta_{\mathrm{l}}^2}{\tau_{r,i}} \sigma^2\right]  \nonumber    \\
&\overset{(b)}{\leq}\! 1.2L^2\bar{\tau}_r^3{\eta_{\rm g}\eta_{\mathrm{l}} ^3}[{G^2} \!+\! {H^2}{\left\| {\nabla F({{\bf{w}}_r})} \right\|^2}] \!+\! \frac{{L^2\bar{\tau}_r^3{\eta_{\rm g}\eta_{\mathrm{l}} ^3}}}{2M_r}\!\sum\limits_{i \in \mathcal{M}_r} \!{\frac{\sigma^2}{{{\tau _{r,i}}}}},
\end{align}
where $e_{r,i}= \frac{1.2\tau_{r,i}-1}{2}\left[ \frac{1-3.2\tau_{r,i}}{2} + \frac{1.2\tau_{r,i}-1}{2} \left(\frac{1.2\tau_{r,i}+1}{1.2\tau_{r,i}-1}\right)^{\tau_{r,i}} \right]
\cdot \\ \left[ 
2.4\frac{\bar{\tau}_r^2}{\tau_{r,i}}\eta_{\mathrm{l}}^2\left\|\nabla f_i({\bf{w}}_r)\right\|^2 + \frac{\bar{\tau}_r^2}{\tau_{r,i}^2}\sigma^2\eta_{\mathrm{l}}^2 \right]$; 
(a) holds due to $\frac{1.2\tau_{r,i}-1}{2\tau_{r,i}^2}\left[ \frac{1-3.2\tau_{r,i}}{2} \!+\! \frac{1.2\tau_{r,i}-1}{2} \left(\frac{1.2\tau_{r,i}+1}{1.2\tau_{r,i}-1}\right)^{\tau_{r,i}} \right] < 1$, $\forall \tau_{r,i} \geq 1$; and (b) follows readily from \textbf{Assumption 3}.

Next, we derive the upper bound of $A_2$ in \eqref{overall 1}, as given by
\begin{align}\label{A2_}
&A_2  = \frac{L}{2}\eta_{\rm g}^2\eta_{\mathrm{l}}^2\mathbb{E}\left[\left\|\sum_{i \in \mathcal{M}_r}\frac{1}{M_r}\sum_{j=0}^{\tau_{r,i}-1}\frac{\bar{\tau}_r}{\tau_{r,i}}g_{r,i}^j\right\|^2\right] \nonumber \\
&\overset{(a)}{=} \frac{L}{2}\eta_{\rm g}^2\eta_{\mathrm{l}}^2\mathbb{E}\left[\left\|\frac{1}{M_r}\!\!\sum_{i\in \mathcal{M}_r}\sum_{j=0}^{\tau_{r,i}-1}\frac{\bar{\tau}_r}{\tau_{r,i}}\left(g_{r,i}^j \!- \! \nabla f_i({\bf{w}}_{r,i}^j)\right) \right\|^2\right]  \nonumber \\
&~~~+ \frac{L}{2}\eta_{\rm g}^2\eta_{\mathrm{l}}^2 \left\|\frac{1}{M_r}\sum_{i \in \mathcal{M}_r}\sum_{j=0}^{\tau_{r,i}-1}\frac{\bar{\tau}_r}{\tau_{r,i}}\nabla f_i({\bf{w}}_{r,i}^j) \right\|^2  \nonumber \\
&\overset{(b)}{\leq} \frac{L}{2}\eta_{\rm g}^2\eta_{\mathrm{l}}^2\frac{1}{M_r^2}\sum_{i \in \mathcal{M}_r}\frac{{\bar{\tau}_r}^2}{\tau_{r,i}^2}\sum_{j=0}^{\tau_{r,i}-1}\sigma^2  \nonumber \\
&~~~+ \frac{L}{2}\eta_{\rm g}^2\eta_{\mathrm{l}}^2 {\bar{\tau}_r}^2\left\|\frac{1}{M_r}\sum_{i \in \mathcal{M}_r}\sum_{j=0}^{\tau_{r,i}-1}\frac{1}{\tau_{r,i}}\nabla f_i({\bf{w}}_{r,i}^j) \right\|^2  \nonumber \\
&\overset{(c)}{\leq} \frac{L\eta_{\rm g}^2\eta_{\mathrm{l}}^2{\bar{\tau}_r}^2}{2M_r^2}\sigma^2\sum_{i \in \mathcal{M}_r}\frac{1}{\tau_{r,i}} + \frac{L{\bar{\tau}_r}^2\eta_{\rm g}^2\eta_{\mathrm{l}}^2}{2\rho^2}\left\| \nabla F({\bf{w}}_r) \right\|^2 ,
\end{align}    
where (a) is due to $\mathbb{E}\left[\|\mathbf{x}\|^2\right]=\mathbb{E}\left[\|\mathbf{x}-\mathbb{E}[\mathbf{x}]\|^2\right]+\|\mathbb{E}[\mathbf{x}]\|^2$; (b) is due to $\mathbb{E}\left[\left\|x_1+\cdots+x_n\right\|^2\right] \leq n\mathbb{E}\left[\left\|x_1\right\|^2+\cdots+\left\|x_n\right\|^2\right]$. 
We can further prove that there exists a constant $\rho>0$ such that
\begin{equation} \label{extended BGD2}
   \left\| \frac{1}{M_r}\sum_{i \in \mathcal{M}_r} \frac{1}{\tau_{r,i}} \nabla f_i(\boldsymbol{\rm{w}}_r)\right\|^2 \leq  \left\| \frac{1}{\rho} \nabla F(\boldsymbol{\rm{w}}_r)\right\|^2 , \forall r, i.
\end{equation}
Applying the Cauchy-Schwartz inequality and \textbf{Assumption 3} on the LHS of \eqref{extended BGD2} yields $\rho^2 \!\leq\! \frac{M_r\left\| \nabla F(\boldsymbol{\rm{w}}_r)\right\|^2}{\sum_{i \in \mathcal{M}_r}\frac{1}{\tau_{r,i}^2}\left(G^2 + H^2\left\| \nabla F(\boldsymbol{\rm{w}}_r)\right\|^2\right)}$. Setting $\tau_{r,i} = 1, \forall r,i$, the upper bound of $\rho$ is $\rho \leq \frac{1}{\sqrt{H^2 + G^2 / \left\| \nabla F(\boldsymbol{\rm{w}}_r)\right\|^2}} $ with a bounded global gradient $\left\| \nabla F(\boldsymbol{\rm{w}}_r)\right\|$, and (c) follows from~\eqref{extended BGD2}.

With the upper bound of $A_1$ based on \eqref{A1_} and \eqref{B1_} and the upper bound of $A_2$ based on \eqref{A2_}, \eqref{overall 1} can be reorganized into \eqref{one round error}, 
where we define a positive constant ${\mathrm{q}}_r$ that satisfies the condition $0 < {\mathrm{q}_{r}}<1-\frac{\bar{\tau}_r \eta_{\rm g} \eta_{\mathrm{l}}}{\rho^2}-2.4 H^2 L^2 \bar{\tau}_r^2 \eta_{\mathrm{l}}^2<1$. 

\section*{Appendix C \\ Proof of Theorem 2}

By adding $F({\bf{w}^{*}})$ to both sides of \eqref{one round error} and then reorganizing and taking expectation, we have
\begin{align} \label{septal terms cancel}
\mathbb{E}\left[\left\| \nabla F({\bf{w}}_r) \right\|^2\right] &\leq a_2 + \frac{2.4}{\mathrm{q}_r}G^2L^2\bar{\tau}^2\eta_{\mathrm{l}}^2 \\ \nonumber
&+ \frac{L\bar{\tau}\eta_{\mathrm{l}}}{\mathrm{q}_r}\left(\frac{ L \bar{\tau} \eta_{\mathrm{l}}}{ M_r} + \frac{\eta_{\rm g}}{M_r^2}\right) \sigma^2 \sum_{i \in \mathcal{M}_r} \frac{1}{\tau_{r, i}},
\end{align}
where $ a_2 = \frac{2\left( \mathbb{E}\left[F({\bf{w}}_r) - F({\bf{w}^*})\right] - \mathbb{E}\left[F({\bf{w}}_{r+1}) -F({\bf{w}^*})\right]\right)}{\bar{\tau}\eta_{\rm g}\eta_{\mathrm{l}} \mathrm{q}_r}$.
To derive an upper bound of $\sum_{i \in \mathcal{M}_r} \frac{1}{\tau_{r,i}}$ in \eqref{septal terms cancel}, we define $\kappa_r \buildrel \Delta \over = \mathop {\max }\limits_{r, i \in \mathcal{M}_r} \left\{ \frac{\bar{\tau}}{\tau_{r,i}} \right\}$, and rewrite \eqref{septal terms cancel} as
\begin{align} \label{septal terms cancel 2}
\mathbb{E}\left[\left\| \nabla F({\bf{w}}_r) \right\|^2\right] &\!\leq \! a_2 \!+ \!\frac{2.4}{\mathrm{q}_r}G^2L^2\bar{\tau}^2\eta_{\mathrm{l}}^2 
\!+\! \frac{L\eta_{\mathrm{l}}}{\mathrm{q}_r}\left(\! L \bar{\tau} \eta_{\mathrm{l}} \!+ \!\frac{\kappa_{r} \eta_{\rm g}}{M_r}\!\right) \sigma^2 .
\end{align}

By summing up \eqref{septal terms cancel 2} from $r = 1,\cdots, R$ and dividing both sides of the resulting inequality by $R$, we obtain
\begin{align} \label{septal terms cancel 3}
&\frac{1}{R}\sum_{r=1}^{R} \mathbb{E}\left[||\nabla F({\bf{w}}_r)|{|^2}\right] \leq \mathbb{E}_{\mathcal{M}}\left[\frac{1}{R}\sum_{r=1}^{R}a_2 \right.\\ \nonumber
& + \frac{1}{R}\sum_{r=1}^{R} \frac{2.4}{\mathrm{q}_r}G^2L^2\bar{\tau}^2\eta_{\mathrm{l}}^2 + \left. \frac{1}{R}\sum_{r=1}^{R} \frac{L\eta_{\mathrm{l}}}{\mathrm{q}_r}\left( L \bar{\tau} \eta_{\mathrm{l}} + \frac{\kappa_{r} \eta_{\rm g}}{M_r}\right) \sigma^2 \right],
\end{align}  
where $\mathbb{E}_{\mathcal{M}}[\cdot]$ takes expectation over the past device scheduling decisions, i.e., $\mathcal{M} = \{\mathcal{M}_r | r = 1, \cdots, R \}$.
Define $\mathrm{q}  \buildrel \Delta \over = \mathop {\min }\limits_r \{ \mathrm{q}_r \}$, $\kappa  \buildrel \Delta \over = \mathop {\max }\limits_r \{\kappa_r\}$, and $M  \buildrel \Delta \over = \mathop {\min }\limits_r \{M_r\}$ with $M \geq 1$, and further rescale the RHS of \eqref{septal terms cancel 3}. We finally obtain
\begin{align}
    &\frac{1}{R}\sum_{r=1}^{R} \mathbb{E}\left[||\nabla F({\bf{w}}_r)|{|^2}\right] \leq \frac{2[F({\bf{w}}_1)-F({\bf{w}}^*)]}{\bar{\tau}\eta_{\rm g}\eta_{\mathrm{l}} \mathrm{q}R} \\ \nonumber
    & + \frac{2.4}{\mathrm{q}}G^2L^2\bar{\tau}^2\eta_{\mathrm{l}}^2 + \frac{L\eta_{\mathrm{l}}}{\mathrm{q}}\left( L \bar{\tau} \eta_{\mathrm{l}} + \frac{\kappa \eta_{\rm g}}{M}\right) \sigma^2 ,
\end{align}
which completes this proof.

\section*{Appendix D \\ Proof of Theorem 3}

For illustration convenience, we suppress the indexes of training rounds and devices, i.e., $r$ and $i$.
The transmission rate $u$ monotonically increases w.r.t $b > 0$, since
\begin{align} \label{derivative}
\frac{\mathrm{d} u}{\mathrm{d} b} =\frac{1}{\ln 2}\left(\ln \left(1+\frac{p h^2}{ b N_0}\right) - \frac{p h^2}{b N_0+p h^2} \right)>0. 
\end{align}

We can prove by contradiction that $t_r^*$ must be equal for all selected devices under the optimal bandwidth allocation.
To see it, we first assume that $\boldsymbol{b}^{'}$ denotes the optimal bandwidth allocation, with $t_{r,i}$ inconsistent across devices, and let the fastest and slowest devices be denoted by $m$, $n$, respectively. Due to synchronous aggregation, the total delay of each round is determined by the slowest device. Then according to \eqref{derivative}, the bandwidth $b_m$ allocated to Device $m$ can be partially reassigned to Device $n$ to increase its transmission rate and in turn decrease the total latency determined by Device $n$. Such an operation can be repeated continuously until a solution $\boldsymbol{b}^*$ of Problem \textbf{P3} is obtained when the latency $t_{r,i}$ of all the selected devices is the same.
The solution $\boldsymbol{b}^*$ clearly leads to a strictly smaller objective in \eqref{obj of P3} than that of ``optimal" $\boldsymbol{b}^{'}$, leading to a contradiction. Having proved that $t_r^*$ is equal, the optimal bandwidth allocation can be obtained by solving the following system of equations:
    \begin{equation}
    \left\{\begin{array}{l}
t_{r, i}^{\mathrm{comp}}+\frac{S}{b_{r,i}^* \log _2\left(1+\frac{p_i h_{r, i}^2}{b_{r,i}^* N_0}\right)}=t_r^*, \forall i \in \mathcal{M}_r \\
\sum_{i \in \mathcal{M}_r} b_{r,i}^* = B .  \\
\end{array}\right.
\end{equation}
By leveraging the definition of Lambert-W function, the theorem readily follows.

\section*{Appendix E \\ Proof of Theorem 4}
Consider the currently selected subset of devices, denoted as $\mathcal{Q}_k$, including $Q$ devices, and the objective function \eqref{obj of P2}. If the introduction of another device $i \in \mathcal{K} \setminus \mathcal{Q}_k$ into $\mathcal{Q}_k$ allows for a further decrease of \eqref{obj of P2}, then we have
\begin{equation}
    \begin{aligned} \label{decline_}
    & \frac{Q+1+\gamma}{(Q+1)^2} \left(\Delta + \frac{1}{\tau_{r,i}}\right) <  \frac{Q+\gamma}{Q^2} \Delta,
\end{aligned}
\end{equation}
where $\Delta = \sum_{q \in \mathcal{Q}_k} \frac{1}{\tau_{r,q}}$. This readily implies \eqref{descent condition}, the proof is thus complete.

\bibliographystyle{IEEEtran}

\bibliography{references}

\begin{thebibliography}{10}
\providecommand{\url}[1]{#1}
\csname url@samestyle\endcsname
\providecommand{\newblock}{\relax}
\providecommand{\bibinfo}[2]{#2}
\providecommand{\BIBentrySTDinterwordspacing}{\spaceskip=0pt\relax}
\providecommand{\BIBentryALTinterwordstretchfactor}{4}
\providecommand{\BIBentryALTinterwordspacing}{\spaceskip=\fontdimen2\font plus
\BIBentryALTinterwordstretchfactor\fontdimen3\font minus \fontdimen4\font\relax}
\providecommand{\BIBforeignlanguage}[2]{{%
\expandafter\ifx\csname l@#1\endcsname\relax
\typeout{** WARNING: IEEEtran.bst: No hyphenation pattern has been}%
\typeout{** loaded for the language `#1'. Using the pattern for}%
\typeout{** the default language instead.}%
\else
\language=\csname l@#1\endcsname
\fi
#2}}
\providecommand{\BIBdecl}{\relax}
\BIBdecl

\bibitem{9349624}
W.~Jiang, B.~Han, M.~A. Habibi, and H.~D. Schotten, ``The road towards 6{G}: A comprehensive survey,'' \emph{IEEE Open J. Commun. Soc.}, vol.~2, pp. 334--366, 2021.

\bibitem{9446488}
S.~Hu, X.~Chen, W.~Ni, E.~Hossain, and X.~Wang, ``Distributed machine learning for wireless communication networks: Techniques, architectures, and applications,'' \emph{IEEE Commun. Surv. Tutorials}, vol.~23, no.~3, pp. 1458--1493, 2021.

\bibitem{8648462}
T.~Chen, S.~Barbarossa, X.~Wang, G.~B. Giannakis, and Z.-L. Zhang, ``Learning and management for internet of things: Accounting for adaptivity and scalability,'' \emph{Proc. IEEE}, vol. 107, no.~4, pp. 778--796, 2019.

\bibitem{konečný2018federated}
J.~Konečný, H.~B. McMahan, F.~X. Yu, A.~T. Suresh, D.~Bacon, and P.~Richtárik, ``Federated learning: Strategies for improving communication efficiency,'' in \emph{Proc. Int. Conf. Learn. Representations (ICLR)}, 2018.

\bibitem{fedavg}
B.~McMahan, E.~Moore, D.~Ramage, S.~Hampson, and B.~A.~y. Arcas, ``Communication-efficient learning of deep networks from decentralized data,'' in \emph{Proc. Int. Conf. Artif. Intell. Stat. (AISTATS)}, 2017, pp. 1273--1282.

\bibitem{9141214}
S.~Niknam, H.~S. Dhillon, and J.~H. Reed, ``Federated learning for wireless communications: Motivation, opportunities, and challenges,'' \emph{IEEE Commun. Mag.}, vol.~58, no.~6, pp. 46--51, 2020.

\bibitem{9169921}
H.~Shiri, J.~Park, and M.~Bennis, ``Communication-efficient massive {UAV} online path control: Federated learning meets mean-field game theory,'' \emph{IEEE Trans. Commun.}, vol.~68, no.~11, pp. 6840--6857, 2020.

\bibitem{9250516}
F.~Yin, Z.~Lin, Q.~Kong, Y.~Xu, D.~Li, S.~Theodoridis, and S.~R. Cui, ``Fedloc: Federated learning framework for data-driven cooperative localization and location data processing,'' \emph{IEEE Open J. Signal Process.}, vol.~1, pp. 187--215, 2020.

\bibitem{ammad2019federated}
M.~Ammad-Ud-Din, E.~Ivannikova, S.~A. Khan, W.~Oyomno, Q.~Fu, K.~E. Tan, and A.~Flanagan, ``Federated collaborative filtering for privacy-preserving personalized recommendation system,'' \emph{\rm arXiv preprint arXiv:1901.09888}, 2019.

\bibitem{nilsson2018performance}
A.~Nilsson, S.~Smith, G.~Ulm, E.~Gustavsson, and M.~Jirstrand, ``A performance evaluation of federated learning algorithms,'' in \emph{Proc. Wksp. Dritrib. Infrastruct. Deep Learn. (DIDL Wksps)}, 2018, pp. 1--8.

\bibitem{li2019convergence}
X.~Li, K.~Huang, W.~Yang, S.~Wang, and Z.~Zhang, ``On the convergence of {F}ed{A}vg on non-iid data,'' \emph{\rm arXiv preprint arXiv:1907.02189}, 2019.

\bibitem{li2020federated}
T.~Li, A.~K. Sahu, A.~Talwalkar, and V.~Smith, ``Federated learning: Challenges, methods, and future directions,'' \emph{IEEE Signal Process. Mag.}, vol.~37, no.~3, pp. 50--60, 2020.

\bibitem{9579038}
B.~Luo, X.~Li, S.~Wang, J.~Huang, and L.~Tassiulas, ``Cost-effective federated learning in mobile edge networks,'' \emph{IEEE J. Sel. Areas Commun.}, vol.~39, no.~12, pp. 3606--3621, 2021.

\bibitem{zhao2018federated}
Y.~Zhao, M.~Li, L.~Lai, N.~Suda, D.~Civin, and V.~Chandra, ``Federated learning with non-iid data,'' \emph{\rm arXiv preprint arXiv:1806.00582}, 2018.

\bibitem{scaffold}
S.~P. Karimireddy, S.~Kale, M.~Mohri, S.~Reddi, S.~Stich, and A.~T. Suresh, ``Scaffold: Stochastic controlled averaging for federated learning,'' in \emph{Proc. Int. Conf. Mach. Learn. (ICML)}.\hskip 1em plus 0.5em minus 0.4em\relax PMLR, 2020, pp. 5132--5143.

\bibitem{li2021model}
Q.~Li, B.~He, and D.~Song, ``Model-contrastive federated learning,'' in \emph{Proc. IEEE/CVF Conf. Comp. Vision Pattern Recognit. (CVPR)}, 2021, pp. 10\,713--10\,722.

\bibitem{fedprox}
T.~Li, A.~K. Sahu, M.~Zaheer, M.~Sanjabi, A.~Talwalkar, and V.~Smith, ``Federated optimization in heterogeneous networks,'' \emph{Proc. Mach. Learn. Syst. (MLSys)}, vol.~2, pp. 429--450, 2020.

\bibitem{fednova}
J.~Wang, Q.~Liu, H.~Liang, G.~Joshi, and H.~V. Poor, ``A novel framework for the analysis and design of heterogeneous federated learning,'' \emph{IEEE Trans. Signal Process.}, vol.~69, pp. 5234--5249, 2021.

\bibitem{fedlin}
A.~Mitra, R.~Jaafar, G.~J. Pappas, and H.~Hassani, ``Linear convergence in federated learning: Tackling client heterogeneity and sparse gradients,'' in \emph{Proc. Adv. Neural Inf. Process. Syst. (NeurIPS)}, vol.~34, 2021, pp. 14\,606--14\,619.

\bibitem{XIAO2024}
B.~Xiao, X.~Yu, W.~Ni, X.~Wang, and H.~V. Poor, ``Over-the-air federated learning: Status quo, open challenges, and future directions,'' \emph{Fundamental Research}, pp. 1--1, 2024.

\bibitem{our}
X.~Yu, B.~Xiao, W.~Ni, and X.~Wang, ``Optimal adaptive power control for over-the-air federated edge learning under fading channels,'' \emph{IEEE Trans. Commun.}, vol.~71, no.~9, pp. 5199--5213, 2023.

\bibitem{8737464}
N.~H. Tran, W.~Bao, A.~Zomaya, M.~N.~H. Nguyen, and C.~S. Hong, ``Federated learning over wireless networks: Optimization model design and analysis,'' in \emph{Proc. IEEE Conf. Comput. Commun. (INFOCOM)}, 2019, pp. 1387--1395.

\bibitem{8761315}
T.~Nishio and R.~Yonetani, ``Client selection for federated learning with heterogeneous resources in mobile edge,'' in \emph{Proc. IEEE Int. Conf. Commun. (ICC)}, 2019, pp. 1--7.

\bibitem{8851249}
H.~H. Yang, Z.~Liu, T.~Q.~S. Quek, and H.~V. Poor, ``Scheduling policies for federated learning in wireless networks,'' \emph{IEEE Trans. Commun.}, vol.~68, no.~1, pp. 317--333, 2020.

\bibitem{9337227}
M.~M. Amiri, D.~Gündüz, S.~R. Kulkarni, and H.~V. Poor, ``Convergence of update aware device scheduling for federated learning at the wireless edge,'' \emph{IEEE Trans. Wireless Commun.}, vol.~20, no.~6, pp. 3643--3658, 2021.

\bibitem{10041216}
C.-H. Hu, Z.~Chen, and E.~G. Larsson, ``Scheduling and aggregation design for asynchronous federated learning over wireless networks,'' \emph{IEEE J. Sel. Areas Commun.}, vol.~41, no.~4, pp. 874--886, 2023.

\bibitem{9207871}
W.~Shi, S.~Zhou, Z.~Niu, M.~Jiang, and L.~Geng, ``Joint device scheduling and resource allocation for latency constrained wireless federated learning,'' \emph{IEEE Trans. Wireless Commun.}, vol.~20, no.~1, pp. 453--467, 2021.

\bibitem{9210812}
M.~Chen, Z.~Yang, W.~Saad, C.~Yin, H.~V. Poor, and S.~Cui, ``A joint learning and communications framework for federated learning over wireless networks,'' \emph{IEEE Trans. Wireless Commun.}, vol.~20, no.~1, pp. 269--283, 2021.

\bibitem{9170917}
J.~Ren, Y.~He, D.~Wen, G.~Yu, K.~Huang, and D.~Guo, ``Scheduling for cellular federated edge learning with importance and channel awareness,'' \emph{IEEE Trans. Wireless Commun.}, vol.~19, no.~11, pp. 7690--7703, 2020.

\bibitem{9796935}
B.~Luo, W.~Xiao, S.~Wang, J.~Huang, and L.~Tassiulas, ``Tackling system and statistical heterogeneity for federated learning with adaptive client sampling,'' in \emph{Proc. IEEE Conf. Comput. Commun. (INFOCOM)}, 2022, pp. 1739--1748.

\bibitem{diao2021heterofl}
E.~Diao, J.~Ding, and V.~Tarokh, ``Hetero{FL}: Computation and communication efficient federated learning for heterogeneous clients,'' in \emph{Proc. Int. Conf. Learn. Representations (ICLR)}, 2021.

\bibitem{qu2022generalized}
Z.~Qu, X.~Li, R.~Duan, Y.~Liu, B.~Tang, and Z.~Lu, ``Generalized federated learning via sharpness aware minimization,'' in \emph{International conference on machine learning}.\hskip 1em plus 0.5em minus 0.4em\relax PMLR, 2022, pp. 18\,250--18\,280.

\bibitem{9237168}
J.~Xu and H.~Wang, ``Client selection and bandwidth allocation in wireless federated learning networks: A long-term perspective,'' \emph{IEEE Trans. Wireless Commun.}, vol.~20, no.~2, pp. 1188--1200, 2021.

\bibitem{yang2021achieving}
H.~Yang, M.~Fang, and J.~Liu, ``Achieving linear speedup with partial worker participation in non-iid federated learning,'' in \emph{Proc. Int. Conf. Learn. Representations (ICLR)}, 2021.

\bibitem{MNIST}
Y.~LeCun, L.~Bottou, Y.~Bengio, and P.~Haffner, ``Gradient-based learning applied to document recognition,'' \emph{Proc. IEEE}, vol.~86, no.~11, pp. 2278--2324, 1998.

\bibitem{cifar}
A.~Krizhevsky, G.~Hinton \emph{et~al.}, ``Learning multiple layers of features from tiny images,'' 2009.

\bibitem{reisizadeh2020fedpaq}
A.~Reisizadeh, A.~Mokhtari, H.~Hassani, A.~Jadbabaie, and R.~Pedarsani, ``Fedpaq: A communication-efficient federated learning method with periodic averaging and quantization,'' in \emph{Proc. Int. Conf. Artif. Intell. Stat. (AISTATS)}.\hskip 1em plus 0.5em minus 0.4em\relax PMLR, 2020, pp. 2021--2031.

\bibitem{9139810}
C.~Wang, X.~Wei, and P.~Zhou, ``Optimize scheduling of federated learning on battery-powered mobile devices,'' in \emph{Proc. IEEE Int. Parallel Distrib. Process. Symp. (IPDPS)}, 2020, pp. 212--221.

\end{thebibliography}

\end{document}